\documentclass[11pt]{article}
\usepackage[utf8]{inputenc}
\usepackage[english]{babel}
\usepackage{amsmath}
\usepackage{amsfonts}
\usepackage{amssymb}
\usepackage{amsthm}
\usepackage{fullpage}
\usepackage{todonotes}
\usepackage{tikz}
\usepackage{enumitem}
\usepackage{mathtools}
\usepackage{thmtools}
\usepackage{thm-restate}
\usepackage{booktabs}
\usepackage{url}

\def\ShowAuthNotes{1}
\ifnum\ShowAuthNotes=1
\newcommand{\authnote}[2]{\ \\ \textcolor{red}{\parbox{0.9\linewidth}{[{\footnotesize {\bf #1:} { {#2}}}]}}\newline}
\else
\newcommand{\authnote}[2]{}
\fi

\newcommand{\du}{\ensuremath{\text{OR} \circ \text{AND} \circ \text{OR}}}
\newcommand{\dt}[2]{\ensuremath{\text{OR}_{#1} \circ \text{AND} \circ \text{OR}_{#2}}}
\newcommand{\dz}[1]{\ensuremath{\text{OR} \circ \text{AND} \circ \text{OR}_{#1}}}
\newcommand{\dm}[3]{\ensuremath{\text{OR}_{#1} \circ \text{AND}_{#2} \circ \text{OR}_{#3}}}

\newcommand{\F}{\{0,1\}}
\newcommand{\Ftwo}{\{0,1\}}
\newcommand{\R}{\mathbb{R}}
\newcommand{\C}{\ensuremath{\mathcal{C}}}
\renewcommand{\H}{\mathrm{H}}

\newcommand{\rk}{\operatorname{rank}}

\newcommand{\IP}{\operatorname{IP}}

\def \F {{\mathbb F}}

\def \MOD {{\sf MOD}}
\def \AND {{\sf AND}}
\def \SUM {{\sf SUM}}
\def \PRODUCT {{\sf PRODUCT}}

\newcommand{\eps}{\varepsilon} 

\newtheorem{theorem}{Theorem}[section]
\newtheorem{lemma}[theorem]{Lemma}

\newtheorem{claim}[theorem]{Claim}
\newtheorem{remark}[theorem]{Remark}
\newtheorem{definition}{Definition}[section]
\newtheorem{openproblem}{Open Problem}[section]

\newcommand{\E}{\mathbb{E}}
\DeclareMathOperator{\out}{out}

\DeclarePairedDelimiter\ceil{\lceil}{\rceil}
\DeclarePairedDelimiter\floor{\lfloor}{\rfloor}
\newcommand{\Cor}{\mathrm{Cor}}
\newenvironment{mypic}{\begin{center}\begin{tikzpicture}[scale=1.3,every node/.style={draw,circle,minimum size=5mm,inner sep=0.5mm},label distance=-1mm]}{\end{tikzpicture}\end{center}}

\begin{document}
\title{Circuit Depth Reductions\footnote{The work of Ryan Williams is supported by NSF CCF-1909429 and CCF-1741615. The work of Alexander~S. Kulikov presented in Section~4 is supported by the RNF grant 18-71-10042.}}
\author{
Alexander Golovnev\thanks{Georgetown University, email: alexgolovnev@gmail.com}
\and
Alexander~S. Kulikov\thanks{Steklov Institute of Mathematics at St. Petersburg and St. Petersburg State University, 
email: alexanderskulikov@gmail.com}
\and
R.~Ryan Williams\thanks{MIT CSAIL \& EECS, email: rrw@mit.edu}
}
\date{}
\maketitle


\begin{abstract}
The best known size lower bounds against unrestricted circuits have remained around $3n$ for several decades. Moreover, the only known technique for proving lower bounds in this model, gate elimination, is inherently limited to proving lower bounds of less than $5n$. In this work, we propose a non-gate-elimination approach for obtaining circuit lower bounds, via certain depth-three lower bounds. We prove that every (unbounded-depth) circuit of size $s$ can be expressed as an OR of $2^{s/3.9}$ $16$-CNFs. For DeMorgan formulas, the best known size lower bounds have been stuck at around $n^{3-o(1)}$ for decades. Under a plausible hypothesis about probabilistic polynomials, we show that $n^{4-\eps}$-size DeMorgan formulas have  $2^{n^{1-\Omega(\eps)}}$-size depth-3 circuits which are approximate sums of $n^{1-\Omega(\eps)}$-degree polynomials over $\F_2$. While these structural results do not immediately lead to new lower bounds, they do suggest new avenues of attack on these longstanding lower bound problems.

Our results complement the classical depth-$3$ reduction results of Valiant, which show that \emph{logarithmic}-depth circuits of linear size can be computed by an OR of $2^{\eps n}$ $n^{\delta}$-CNFs, and slightly stronger results for series-parallel circuits. It is known that no purely graph-theoretic reduction could yield interesting depth-3 circuits from circuits of super-logarithmic depth. We overcome this limitation (for small-size circuits) by taking into account both the graph-theoretic and functional properties of circuits and formulas.

We show that improvements of the following pseudorandom constructions imply super-linear circuit lower bounds for log-depth circuits via Valiant's reduction: dispersers for varieties, correlation with constant degree polynomials, matrix rigidity, and hardness for depth-$3$ circuits with constant bottom fan-in. On the other hand, our depth reductions show that even modest improvements of the known constructions give elementary proofs of improved (but still linear) circuit lower bounds.

%
%
\end{abstract}

\thispagestyle{empty}
\addtocounter{page}{-1}
\newpage    


\section{Introduction}

The Boolean circuit model is natural for computing Boolean functions. A~circuit corresponds to a~simple straight line program where every instruction performs a~binary operation on two operands, each of which is either an input or the~result of a~previous instruction. The structure of this program is extremely simple: no loops, no conditional statements. Still, we know no functions in P (or even NP, or even E$^{\text{NP}}$) that requires even $3.1n$ binary instructions (``size'') to compute on inputs of length $n$. 
This is in sharp contrast with the fact that it is easy to \emph{non-constructively} find such functions: simple counting arguments show 
a~random function on~$n$ variables has circuit size $\Omega(2^n/n)$ with probability $1-o(1)$~\cite{Shannon49}. 

The strongest known circuit size lower bound $(3+\frac{1}{86})n-o(n)$ was proved for affine dispersers for sublinear dimension~\cite{FGHK16}. This proof, as well as all previous proofs for general circuit lower bounds against explicit functions, is based on the method of gate elimination. The main idea is to find a~substitution to an input variable that eliminates sufficiently many gates from the given circuit, and then proceed by induction. While this is the most successful method known so far for proving lower bounds for unrestricted circuits, the resulting case analysis becomes increasingly tedious: when eliminating (say) $3$ or $4$ gates, one must consider all possible cases when two of these gates coincide. It is difficult to imagine a~proof of~$5n$ lower bound using these ideas. This intuition was recently made formal in~\cite{GHKK18}, where it was shown that a~certain formalization of the gate elimination technique is unable to obtain a~stronger than~$5n$ lower bound. Therefore we must find new approaches for proving lower bounds against circuits of unbounded depth. Let us review some of the prior results on various circuit models.


\paragraph{Linear Circuits.}
Superlinear lower bounds are not known even for linear circuits, i.e., circuits consisting of only XOR gates (also known as $\oplus$ gates). Note every linear function with one output has a circuit of size at most $n-1$. For linear circuits, we consider \emph{linear transformations}, multi-output functions of the form 
$f(x)=Ax$ where $A \in \F_2^{m \times n}$. For a~random matrix~$A \in \{0,1\}^{n \times n}$, the size of the smallest linear circuit computing~$Ax$ is $\Theta(n^2/\log{n})$~\cite{L56} with probability $1-o(1)$, but for explicitly-constructed matrices the strongest known lower bound is $3n-o(n)$ due to Chashkin~\cite{chashkin}. Interestingly, Chashkin's proof is not based on gate elimination: he first shows that the parity check matrix~$H \in \{0,1\}^{\log n \times n}$ of the Hamming code has circuit size $2n-o(n)$ by proving that every circuit for $H$ has at least $n-o(n)$ gates of out-degree at least~$2$.\footnote{All logarithms are base $2$ unless noted otherwise.} Then he ``pads'' $H$ to an $n \times n$ matrix $H'$ and shows that $n-o(n)$ additional gates are needed for $H'$. Similarly, the best known lower bound on the complexity of linear circuits with $\log{n}\leq m<o(n^2)$ outputs is $2n+m-o(n)$ (also follows from~\cite{chashkin}).

\paragraph{Log-Depth Circuits.}
Nothing stronger than a $(3+\frac{1}{86})n-o(n)$ size lower bound is known even for circuits of depth $O(\log n)$. It is straightforward to show that any function that depends on all of its $n$~variables requires depth at least~$\log n$. One can also present an~explicit function that cannot be computed by a~circuit of depth smaller than $2\log n-o(\log n)$ using Nechiporuk's lower bound of $n^{2-o(1)}$ on formula size over the full binary basis~\cite{Nec61}. Still, proving superlinear size lower bounds for circuits of depth $O(\log n)$ remains a~major open problem~\cite{V77}.

\paragraph{Constant-Depth Circuits.}
Another natural and simple model of computation is bounded-depth unbounded fan-in circuits, which correspond to highly parallelizable computation. In this paper, we focus on depth-2 circuits of the form $\text{AND} \circ \text{OR}$ (i.e., CNFs) and depth-3 circuits of the form {\du} (i.e., ORs of CNFs), where the inputs of the circuit are variables and their negations, and the gates have unbounded fan-in. Such circuits are much more structured, and therefore are easier to analyze and to prove lower bounds. For example, it is easy to show that the minimal number of clauses in a~CNF computing the parity of~$n$ bits is equal to $2^{n-1}$, which yields an optimal lower bound for depth-$2$ circuits. However, already for depth~3 there is a~large gap between known lower and upper bounds: it is known~\cite{danc,serg} that the minimum depth-3 circuit size of a~random function on $n$~variables is $\Theta(2^{n/2})$, but the best known lower bound for an explicit function is $2^{\Omega(\sqrt{n})}$~\cite{hastad1986almost,hastad1993top,PPZ97,B97,PPSZ05,MW17}.

Much stronger lower bounds are known for depth-3 circuits where the fan-in of the ``bottom'' gates (those closest to the inputs) is bounded by a parameter $k$. Namely, for any $k\leq O(\sqrt{n})$, Paturi, Saks, and Zane~\cite{PPZ97} proved a $2^{n/k}$ lower bound for computing parity, Wolfovitz~\cite{wolf} proved a lower bound of $(1+1/k)^{n+O(\log{n})}$ for $\mathrm{ETHR}_{\frac{n}{k+1}}$~\footnote{$\mathrm{ETHR}_{\frac{n}{k+1}}$ outputs $1$ if and only if the sum of the $n$ input bits over the integers equals $\frac{n}{k+1}$.}, and a~stronger lower bound of $2^{\frac{\mu_k n}{k-1}}$ for $k\geq 3$ and some constants $\mu_k>1$ was proven in~\cite{PPSZ05} for a~BCH code. 
For example,~\cite{PPSZ05} gives a~lower bound of $2^{0.612n}$ when the bottom fan-in of the circuit is $k=3$, and
a~lower bound of $2^{n/10}$ for the bottom fan-in $k=16$.
For the case of bottom fan-in $k=2$, even a~$2^{n-o(n)}$ lower bound is known~\cite{PSZ97}. 

A simple counting argument shows that for any constant $k=O(1)$, a random function requires depth-$3$ circuits of size $2^{n-o(n)}$.
Calabro, Impagliazzo, and Paturi~\cite{calabro2006duality} construct a family of $2^{O(n^2)}$ explicit functions, most of which require depth-$3$ circuits with $k=O(1)$ of size $2^{n-o(n)}$. Santhanam and Srinivasan~\cite{santhanam2012limits} improve on this by constructing such a family of functions of size $2^{f(n)}$ for every $f(n)=\omega(n\log{n})$.



\paragraph{DeMorgan Formulas.}
While explicit super-linear lower bounds for \emph{circuits} are not known, there are super-linear lower bounds for \emph{formulas}.
In this paper, we focus on the well-studied DeMorgan formulas, which are circuits where every intermediate computation is used exactly once: all gates have out-degree one, and the operations are fan-in two ANDs and ORs, with inputs being variables and their negations.
The two most successful methods for proving lower bounds on DeMorgan formula size are random restrictions
~\cite{s61,a87,in93,pz93,h98,t14} as well as Karchmer--Wigderson games and the Karchmer--Raz--Wigderson conjecture~\cite{k71, DBLP:journals/siamdm/KarchmerW90, DBLP:journals/cc/KarchmerRW95, Gavinsky:2014:TBF:2591796.2591856, dinur_et_al:LIPIcs:2016:5841}. Both approaches have led to a~lower bound of $n^{3-o(1)}$ and are currently stuck at giving stronger lower bounds. 

\subsection{Valiant's Depth Reduction}
Remarkably, a classical result of Valiant from the 70's relates three of the four models above: linear, log-depth, and constant-depth circuits. Using a depth reduction for DAGs~\cite{erdos1975sparse}, Valiant~\cite{V77} shows that for any circuit of size~$cn$ and depth~$d$, and for every integer~$k$, one can remove at most $\frac{2ckn}{\log{d}}$ wires such that the resulting circuit has depth at most $d/2^k$. 
Letting $k$ be a sufficiently large constant, this wire-removal lemma shows how any circuit of size $O(n)$ and depth $O(\log n)$ can be converted into an {\du} circuit where the OR output gate has fan-in $2^{O(n/\log\log n)}$ and the lower OR~gates have fan-in $O(n^{\eps})$ for any desired $\eps > 0$. Hence, by exhibiting a function that has no depth-3 circuit with these restrictions, it follows that this function cannot be computed by circuits of linear size and logarithmic depth. Unfortunately, the best known lower bounds on depth-3 circuits (as mentioned earlier) are still too far from those required for this reduction.

In the same paper, Valiant introduced the notion of matrix rigidity (a similar notion was independently introduced by Grigoriev~\cite{grigor1976application}) and related it to the size of linear circuits of log-depth using ideas similar to those described above. Alas, the known lower bounds on matrix rigidity are also far from being able to give new lower bounds on the size of log-depth linear circuits.

\subsection{Our Results: New Depth Reductions}
The main contributions of this paper are new reductions to depth-$3$ circuits that work for \emph{unrestricted} circuits and (conditionally) for super-cubic formulas, as well as new results connecting various pseudorandom objects to circuit lower bounds.  In particular, we show how to express super-cubic DeMorgan formulas as subexponential-size depth-3 circuits of a certain form, under the hypothesis that DeMorgan formulas have probabilistic polynomials of non-trivial degree. This suggests an approach for improving formula size lower bounds, by proving strong lower bounds on depth-3 circuits.

\subsubsection{Depth Reductions for Circuits}

In Valiant's depth reduction, one can only have $d/2^k < \log{n}$ (and $< cn$ removed edges) for circuits of depth $d\leq O(\log{n})$. Thus, Valiant's depth reduction technique does not yield interesting results for circuits of super-logarithmic depth. Moreover, Schnitger and Klawe~\cite{schnitger1982family,schnitger1983depth,klawe1994shallow} construct an explicit family of DAGs showing that the parameters achieved by Valiant are essentially optimal. Their counterexamples convincingly show that a pure graph-theoretic approach to circuit depth reduction cannot give non-trivial results for unrestricted circuits.

In this paper, we overcome this difficulty by presenting a~counterpart of Valiant's depth reduction that works for circuits of unrestricted depth. Our depth reduction takes into account not only the underlying graph of a~circuit, but also the \emph{functions} computed by the circuit gates.

Our first result shows that unbounded-depth circuits of size less than $3.9n$ can be converted into $2^{\delta n}$ disjunctions of short 16-CNFs, for some $\delta < 1$. 

\begin{restatable}{theorem}{maintransformation}
\label{thm:maintransformation}
Every~circuit of size~$s$ can be computed as
an \dm{2^{\ceil{\frac s2}}}{s}{2} circuit
and as
an \dm{2^{\ceil{\frac s{3.9}}}}{2^{14} \cdot s}{16} circuit.
\end{restatable}

As a consequence, in order to prove a~$3.9n-o(n)$ size lower bound on \emph{unrestricted} circuits, it suffices to provide a function that
cannot be computed by an OR of fewer than $2^{n-o(n)}$ 16-CNF's.
To prove Theorem~\ref{thm:maintransformation}, we gradually transform the given circuit into an OR of CNF's by carefully picking a~suitable internal gate and branching on its two possible output values. In contrast to Valiant's reduction, our transformation works for circuits of arbitrary depth. This is achieved by an argument that takes into account both the graph structure of the circuit \emph{and} the functional properties of the gates involved. Since in this approach we can branch on \emph{internal} gates (inside the circuit), we can avoid 
a~massive case analysis. This also distinguishes our approach from known circuit lower bound proofs based on gate elimination, which must set input gates (or gates very close to the inputs) for the argument to work.

It should be noted that known satisfiability algorithms based on branching, as well as circuit lower bounds based on gate elimination~\cite{PPZ97,PPSZ05,schuler2005algorithm,S10,ck15} may be viewed as depth-reductions for small circuits: if at most $k$ variables are set in any branch before the circuit has a ``trivial'' form, then the circuit can be expressed as an OR of $2^k$ ``trivial'' forms. At the same time, the known techniques in this line of work appear stuck at lower bounds of around $3n$, and provably cannot go beyond linear-size bounds~\cite{GHKK18}.

On the way to proving Theorem~\ref{thm:maintransformation}, we  study structural results about converting small circuits into disjunctions of $k$-CNFs, that have curious connections to properties of $k$-CNFs found in the Satisfiability Coding Lemma~\cite{PPZ97, PPSZ05} and Sparsification Lemma~\cite{IPZ,calabro2006duality}. In particular, we ask the following question.

\begin{restatable}{openproblem}{ouropenproblem}
Prove or disprove: for any constant~$c$, any circuit of size~$cn$ can be computed as an 
\[\dt{2^{(1-\delta(c))n}}{\gamma(c)}\] circuit, for some $\delta(c)>0$ and integer $\gamma(c) \geq 1$.
\end{restatable}

If such depth-3 circuits always existed, this would constitute a new approach to proving superlinear circuit lower bounds. If no depth-3 circuit of this form exists for some linear-size circuits, then we would have a separation between linear-size circuits and (for example) super-linear-size series-parallel circuits (by Valiant's reduction for such circuits, see Theorem~\ref{thm:valiant}). 
Note that for the gate elimination method such limitations are known~\cite{GHKK18}, and they do not apply to the approach presented in this work. 

Our second result  is a new ``non-rigidity'' result for matrices with small linear circuits: if a~matrix $M$ over $\F_2$ can be computed by a~linear circuit of size~$s$, then it is possible to flip at most 16 bits in every row of $M$ to drop its rank below $s/4$. This opens up an approach to proving linear circuit lower bounds on sizes up to $4n$.

\begin{restatable}{theorem}{lineartransformation}
\label{thm:linear}
For every~matrix $M\in\F_2^{m\times n}$ of linear circuit complexity~$s$,
%
$\R_{M}(\floor{s/4}) \leq 16 \,.$
\end{restatable}

\subsubsection{Pseudorandom Objects and Circuit Lower Bounds}
The classical result by Valiant shows that improvements of known depth-3 circuit lower bounds and rigid matrices imply super-linear log-depth circuit lower bounds. Our depth reductions show that even modest improvements of the known constructions also
give modest improvements of unrestricted circuit lower bounds.

In Section~\ref{sec:app}, we show that Valiant's and our reduction are applicable to two more types of pseudorandom objects: dispersers for varieties, and functions having small correlation with low degree polynomials. These implications are briefly summarized\footnote{In this table we only present strongest implications from the strongest premises. Our reductions would still give new circuit lower bounds even from weaker objects (see Section~\ref{sec:app} for formal statements of the results). For example, the second line of the table says that a lower bound of $2^{n-o(n)}$ against depth-$3$ circuits would give a lower bound of $3.9n$. On the other hand, a lower bound of $2^{0.8n}$ would lead to an elementary proof of a lower bound of $3.1n$.} in Table~\ref{table:comp}.

%


\newcommand{\wnloglog}{\ensuremath{\omega\left(\frac{n}{\log\log n}\right)}}

\begin{table}[ht]
\begin{center}
\begin{tabular}{llll}
\toprule
& improving known lower bound & to lower bound & implies lower bound\\
\midrule
V & $s_3^{n^\eps}(f) \ge 2^{n^{1-\eps}}$~\cite{PPZ97} & $s_3^{n^\eps}(f) \ge 2^{\wnloglog}$ & $s_{\log}(f)=\omega(n)$\\[1mm]
* & $s_3^{16}(f) \ge 2^{\frac{n}{10}}$~\cite{PPSZ05} & $s_3^{16}(f) \ge 2^{n-o(n)}$ & $s(f)\ge 3.9n$\\[1mm]
\midrule
V & $\left(n^{\eps},\infty,2^{n-n^{1/2-\eps}}\right)$-disp.~\cite{R16} & $\left(n^{\eps},\infty,2^{n-\wnloglog}\right)$-disp. & $s_{\log}(f)=\omega(n)$\\[1mm]
* & $\left(16,\infty,2^{(1-\eps) n}\right)$-disp.~\cite{VW08} & $(16,1.3n,2^{o(n)})$-disp. & $s(f)\ge 3.9n$\\[1mm]
* & $\left(16,\frac{n}{(\log{n})^c},2^{o(n)}\right)$-disp.~\cite{CohenT15} & $(16,1.3n,2^{o(n)})$-disp. & $s(f)\ge 3.9n$\\[1mm]
\midrule
V & $\R_M\left(\wnloglog\right)>\log\log{n}$~\cite{friedman1993note} & $\R_M\left(\wnloglog\right)>n^{\eps}$ & $s_{\oplus,\log}(M) =\omega(n)$\\[1mm]
* & $\R_{M}(\frac n{65}) > 16$~\cite{pudlak1991computation} & $\R_{M}(n-o(n)) > 16$ & $s_{\oplus}(M) \ge 4n$\\[1mm]
\bottomrule
\end{tabular}
\caption{Comparing the depth reductions of this paper (labeled with~*) with the depth reduction of Valiant~\cite{V77} (labeled with~V). We use the following notation (all formal definitions are given in Sections~\ref{sec:def} and~\ref{sec:app}): $s(f)$~is the smallest size of a~circuit computing~$f$, $s_{\log}$ refers to circuits of depth $O(\log n)$, $s_3^k$ refers to circuits that are ORs of $k$-CNFs, $s_{\oplus}$ refers to circuits consisting of $\oplus$ gates only; $(d,m,s)$-disp. stands for a~$(d,m,s)$-disperser, a~function that is not constant on any subset of the Boolean hypercube of size at least~$s$ that is defined as the set of common roots of at most~$m$ polynomials of degree at most~$d$; 
$\R_M(r)$ is the row-rigidity of~$M$ for the rank~$r$ over $\F_2$, i.e., the smallest row-sparsity of a~matrix~$A$ such that $\rk(M \oplus A) \le r$.}
\label{table:comp}
\end{center}
\end{table}

\subsubsection{Depth Reductions for Formulas}
For DeMorgan formulas we give a conditional depth-reduction (stated informally, see Theorem~\ref{thm:main formulas} for a~formal statement): if there is an $\eps > 0$ such that DeMorgan formulas of size~$s$ have probabilistic polynomials of degree $s^{1-\eps}$ and error $1/3$ over $\F_2$, then for some $\delta > 0$ every DeMorgan formula of size $O(n^{3+\delta})$ can be written as an approximate sum of $2^{n^{1-\gamma}}$ degree-$n^{1-\gamma}$ $\F_2$-polynomials for a~constant $\gamma>0$.\footnote{Similar results can be stated for $\F_p$ where $p$ is any prime.} Moreover, if there are probabilistic polynomials of degree $O(\sqrt{s})$ for DeMorgan formulas of size $s$ (which we conjecture is true), our depth reduction holds for DeMorgan formulas of size $n^{3.99}$.

Interestingly, the techniques used to express DeMorgan formulas as depth-3 circuits are totally different from those used in Theorem~\ref{thm:maintransformation} and \ref{thm:linear}.
Namely, we first balance a~formula (without increasing its size too much), decompose it into a~small top part and several small bottom formulas, approximate the top part by a real-valued low-degree polynomial, then rewrite the bottom parts as probabilistic polynomials (as hypothesized). Finally, we collapse these two polynomials into a~depth-3 circuit.

The hypothesis that lower-degree probabilistic polynomials exist for every DeMorgan formula of size $s$ looks very plausible. We have not found an example of a size-$s$ formula that resists the construction of an $O(\sqrt{s})$-degree probabilistic polynomial. Note that such polynomials do exist in the real-approximation sense~\cite{r11}. For example, every symmetric function (such as MAJORITY) has probabilistic polynomials of $O(\sqrt{s})$ degree~\cite{aw15}, and it is not hard to show that the layered OR-AND tree of depth $\log_2(s)$ has a probabilistic polynomial of $O(\sqrt{s})$ degree as well; in fact, \emph{any} layered tree of depth $\log_2(s)$ with the same gate type at each layer (AND or OR) has such degree.\footnote{Briefly: we can always write such formulas as either an OR of ANDs of $O(\sqrt{s})$ literals, or an AND of ORs of $O(\sqrt{s})$ literals. From there, we can simply replace the output gate with an $O(1)$-degree probabilistic polynomial (as in Razborov~\cite{R87}), and the other gates with exact polynomials of $O(\sqrt{s})$ degree.} It is possible that there are ``nasty'' formulas that resist lower-degree probabilistic polynomials, but given the examples we already know, we do not know what they might look like.

\begin{restatable}{openproblem}{DeMorganopenproblem}
Prove or disprove: every DeMorgan formula of size $s$ has a probabilistic polynomial over $\F_2$ of degree $O(\sqrt{s})$ with constant error less than $1/2$. 
\end{restatable}

\subsection{Motivating Example}
\label{sec:motivating}

Here we provide a~simple example of a reduction of unbounded circuits to depth-3 circuits, to give an idea of what is possible. 

A \emph{formula} is a circuit where every internal gate (i.e. not the inputs and not the output) has out-degree exactly $1$. In our simple example, we will show that a circuit of size, say, $2.7n$ can be computed by an OR of $2^{0.9n}$ formulas of small size ($2.7n$). Since we know almost-quadratic lower bounds~\cite{Nec61} on formula size, we may hope to find a function which is not computable by an OR of $\ll 2^n$ linear-size formulas.

\begin{lemma}[Toy Example] \label{lemma:toy}
Every~circuit of size~$s$ can be expressed as an OR of $2^{\ceil{s/3}}$ formulas, each of size less than~$s$.
\end{lemma}

\begin{proof} For a circuit $C$, let $s(C)$ denote its size.
For $s\leq3$, we just transform a~circuit into a~single formula of the same size. For $s>3$, we proceed by induction. If the given circuit $\C$~is a~formula, no transformation is needed. Otherwise take the topologically first gate~$G$ of out-degree at least~2. Note $G$ is computed by a~formula (all previous gates have out-degree $1$); let $t=s(G)$ be the size of this formula. Consider two minimum-size circuits $\C_0$ and $\C_1$ that compute the same function as~$\C$ on the input sets
$\{x \in \{0,1\}^n \colon G(x)=0\}$ and $\{x \in \{0,1\}^n \colon G(x)=1\}$, respectively. We claim that $s(\C_0), s(\C_1) \le s - t - 2 \le s-3$, since to compute $\C_0$ and $\C_1$ one can remove the subcircuit in $C$ computing gate~$G$ as well as two successors of~$G$. The successors can be removed because $G$~outputs a~constant on both parts of the considered partition of the Boolean hypercube, and all gates in the subcircuit of~$G$ are only needed to compute~$G$ ($G$ is computed by a~formula).
Now, note that
$$\C(x)\equiv (\neg G(x) \wedge \C_0(x)) \vee (G(x) \wedge \C_1(x))\, .$$
Applying the induction hypothesis to~$\C_0$ and $\C_1$, we can rewrite~$\C$ as an OR of at most $2^{\ceil{(s-3)/3+1}}\leq 2^{\ceil{s/3}}$ formulas of size $(s-t-2)+(t+1)<s$.
\end{proof}

This result would imply a~circuit lower bound of $3n-o(n)$ for any function that has correlation at most $2^{-n+o(n)}$ with all formulas of linear size. While we do know functions that have exponentially small correlation $2^{-\eps n}$ with formulas of linear size~\cite{S10, KLP12, ST13, KRT13, t14, IK17}, none of them gives a~bound of $2^{-n+o(n)}$. At any rate there is an inherent limitation for this toy approach. By Parseval's identity, every Boolean function has a~Fourier coefficient $\geq 2^{-n/2}$. This implies that the correlation of this function with the corresponding parity function is at least $2^{-n/2}$ (and this is essentially tight correlation with small formulas for a~random function). Since every parity on a subset of inputs can be computed by a~formula of size $\leq n$, Lemma~\ref{lemma:toy} would only be able to prove circuit lower bounds of $1.5n$.

In order to prove stronger circuit lower bounds, we need to improve both parameters: the constant~$3$ in the exponent, and the class of formulas we reduce circuits to. Our Theorem~\ref{thm:maintransformation} achieves this: it reduces a~circuit to an~OR of $2^{\ceil{\frac{s}{3.9}}}$ formulas, each of which is a~16-CNF. Therefore strong enough correlation bounds against 16-CNFs would yield new circuit lower bounds.

\section{Definitions and Preliminaries}\label{sec:def}
\subsection{Unrestricted Circuits}
Let $B_{n,m}$ be the set of all Boolean functions $f \colon \{0,1\}^n \to \{0,1\}^m$ and let $B_2=B_{2,1}$. A~\emph{circuit} is a~directed acyclic graph that has $n$~nodes of in-degree~0 labeled with~$x_1, \dotsc, x_n$ that are called \emph{input gates}. All other nodes are called \emph{internal gates}, have in-degree~2, and are labeled with operations from~$B_2$. Some $m$~gates are also marked as output gates. Such a~circuit computes a~function from $B_{n,m}$ in a~natural way. The \emph{size}~$s(\C)$ of a~circuit~$\C$ is its number of \emph{internal} gates. This definition extends naturally to functions: $s(f)$ is the smallest size of a~circuit computing the function~$f$. 

The \emph{depth} of a~gate~$G$ is the maximum number of edges (also called \emph{wires}) on a~path from an input gate to~$G$. The depth of a~circuit is the maximum depth of its gates. By $s_{\log n}(f)$ we denote the smallest size of a~circuit of depth $O(\log n)$ computing~$f$.

A~circuit is called \emph{linear} if it consists of $\oplus$ gates only. The corresponding circuit size measure is denoted by $s_{\oplus}$.

Our unrestricted circuits are usually drawn with input gates at the top, so by a~top gate of a~circuit we mean a~gate that is fed by two variables. 

\subsection{Series-Parallel Circuits}
A~\emph{labeling} of a~directed acyclic graph $G=(V,E)$ is a~function $\ell\colon V\rightarrow\mathbb{N}$ such that for every edge $(u,v)\in E$ one has $\ell(u)<\ell(v)$. A~graph/circuit~$G$ is called \emph{series-parallel} if there exists a~labeling~$\ell$ such that for no two edges $(u,v), (u',v')\in E$, $\ell(u)<\ell(u')<\ell(v)<\ell(v')$.
The corresponding circuit complexity measure is $s_{\text{sp}}$.

\subsection{Depth-3 Circuits}
Unlike unrestricted circuits, depth-3 circuits are usually drawn the other way around, i.e., with the output gate at the top. In this paper, we focus on {\du} circuits, i.e., ORs of CNFs. We will use subscripts to indicate the fact that the fan-in of a~particular layer is bounded. Namely, an~\dm{p}{q}{r} circuit is an OR of at most~$p$ CNFs each of which contains at most $q$~clauses and at most~$r$ literals in every clause. Since the gates of a~depth~3 circuit are allowed to have an unbounded fan-in, it is natural to define the size of such a~circuit as its number of wires. 
It is not difficult to see that for $k=O(1)$ the size of an~\dz{k} circuit is equal to the fan-in of its output gate up to a~polynomial factor in~$n$. 
By $s_3^k(f)$ we denote the smallest size of an~\dz{k} circuit computing~$f$.

\subsection{Rigidity}
We say that a~matrix $M\in\F_2^{m\times n}$ is \emph{$s$-sparse} if each \emph{row} of $M$ contains at most $s$~non-zero elements.
The \emph{rigidity} of a~matrix $M\in\F_2^{m\times n}$ for the rank parameter $r$ is the minimum sparsity of a~matrix~$A \in \{0,1\}^{m \times n}$ such that $\rk_{\mathbb{F}_2}(M\oplus A)\leq r$:
\[
\R_M(r) = \min\{s \colon \rk_{\mathbb{F}_2}(M \oplus A)\leq r,\; A \text{ is $s$-sparse} \}\,.
\]

\subsection{Probabilistic, Approximate, and Robust Polynomials}
Since even functions of small circuit and formula complexity may only have large-degree polynomial representations, it often proves convenient to use randomized polynomials or polynomials which approximate (rather than exactly compute) a given function.
\begin{definition}[Probabilistic polynomials]
Let $f\colon\Ftwo^n\to\Ftwo$ be a Boolean function. A distribution $\cal D$ of $n$-variate degree-$d$ polynomials over $\mathbb{F}_2$ is a \emph{probabilistic polynomial} for $f$ with degree $d$ and error $\eps$ if for every $x\in\Ftwo^n$,
\[\Pr_{p\sim \cal{D}}[f(x)=p(x)] \geq 1-\eps.
\]
\end{definition}

\begin{definition}[Approximate Polynomials]
Let $f\colon\Ftwo^n\to\Ftwo$ be a Boolean function. An $n$-variate multilinear degree-$d$ polynomial $p$ over $\R$ is an \emph{approximate polynomial} for $f$ with degree $d$ and error $\eps$ if for every $x\in\Ftwo^n$,
\[|p(x)-f(x)|\leq \eps.
\]
\end{definition}

\begin{definition}[Robust Polynomials]
Let $f\colon\Ftwo^n\to[0,1]$ be a polynomial over $\R$. Then a polynomial $p\colon\R^n\to\R$ is \emph{$\delta$-robust} for $f$ if for every $x\in\Ftwo^n$ and for every $\eps\in[-1/3,1/3]^n$,
\[|f(x)-p(x+\eps)|\leq \delta.
\]
\end{definition}

\subsection{Valiant's Depth Reductions}

Here we formally recall the classical depth reduction results by~Valiant~\cite{V77}.

\begin{theorem}[\cite{V77,C08,V09}]\label{thm:valiant}
For every $c\geq 1$ and $\eps>0$ there exists a $\delta>0$ such that every circuit~{\C} of size $cn$ and depth $c\log{n}$ can be computed as 
\begin{enumerate}
\item an $\dt{2^{\frac{\delta n}{\log\log{n}} }}{n^{\eps}}$ circuit
\item and as an $\dt{2^{\eps n}}{2^{(\log{n})^{1-\delta}}}$ circuit.
\end{enumerate}
Furthermore, for every $c \geq 1$ and $\eps  > 0$ there is a $k \geq 1$ such that every series-parallel circuit of size $cn$ and unbounded depth can be computed as an
$
\dt{2^{\eps n}}{k}
$
circuit.
\end{theorem}

Theorem~\ref{thm:valiant} applied to linear circuits yields the following.

\begin{theorem}[\cite{V77,C08,V09}]\label{thm:valiant_linear}
Let $M\in\F^{m \times n}$ be a matrix. For every $c\geq 1$ and $\eps>0$ there exists $\delta>0$ such that, if a~linear circuit {\C} of size $cn$ and depth $c\log{n}$ computes $Mx$ for every $x\in\F^n$, then
\begin{enumerate}
\item $\R_M\left(\frac{\delta n}{\log\log{n}}\right) \leq  n^{\eps}$;
\item and  $\R_M(\eps n) \leq 2^{(\log{n})^{1-\delta}}$.
\end{enumerate}
Furthermore, for every $c \geq 1$ and $\eps  > 0$ there is a $k \geq 1$ such that if~{\C} is a~series-parallel linear circuit of size $cn$ and unbounded depth, then
$
\R_M(\eps n)\leq k \, .
$
\end{theorem}

\section{Formula Depth Reduction}\label{sec:formula reductions}
In this section, we give a (conditional) depth reduction for DeMorgan formulas. We start by balancing a given formula. For this we use the following result due to Tal~\cite{t14}. 
\begin{lemma}[Claim VI.2 in~\cite{t14}]
\label{lem:avishay}
Let $F$ be a DeMorgan formula of size $s$ over the set of variables $X = \{x_1,...,x_n\}$, and $t$ be some parameter; then, there exist $k \leq 36s/t$ formulas
over $X$, denoted by $T_1,...,T_k$, each of size at most
$t$, and there exists a read-once formula $F'$ of size
$k$ such that $F'
(T_1(x),...,T_k(x)) = F(x)$ for all
$x \in\Ftwo^n.$
\end{lemma}

Below we will also make use of the following results by Reichardt~\cite{r11} and Sherstov~\cite{s12}.
\begin{theorem}[\cite{r11}]
\label{thm:reichardt}
If $f\colon\Ftwo^n\to\Ftwo$ can be computed by a DeMorgan formula of size $s$, then $f$ has an approximate polynomial of degree $O(\sqrt{s})$ with error $\eps=1/10$.
\end{theorem}

\begin{theorem}[\cite{s12}]
\label{thm:sherstov}
If $f\colon\Ftwo^n\to[0,1]$ is a polynomial of degree $d$ over $\R$, then there is a $\delta$-robust polynomial $p$ for $f$ of degree $O(d+\log(1/\delta))$.
\end{theorem}

Now we are ready to present the main result of this section: Assuming DeMorgan formulas of size $s$ have probabilistic polynomials of degree $O(s^{1-\delta})$ for some $\delta > 0$, we will obtain subexponential-size depth-3 circuits computing formulas of super-cubic size. 

In the following, a $\SUM$ gate will compute an \emph{approximate sum}: a (real-weighted) sum of the inputs such that, over all Boolean inputs, the sum is within $\pm 1/3$ of the 0-1 value of a desired Boolean function.

\begin{theorem}
\label{thm:main formulas}
Suppose for some $\delta>0$, DeMorgan formulas of size $\ell$ have probabilistic polynomials of degree $\ell^{1-\delta}$ with error $1/3$. Then for every $\alpha<\delta/(1-\delta)$ there is a $\gamma>0$, so that for every formula $F$ of size $s=O(n^{3+\alpha})$, there is a $2^{n^{1-\gamma}}$-size approximate sum of degree-$n^{1-\gamma}$ $\F_2$-polynomials computing $F$. That is, $F$ can be computed by a
\[\SUM_{2^{n^{1-\gamma}}} \circ \MOD2_{2^{n^{1-\gamma}}} \circ \AND_{n^{1-\gamma}} \;.
\]
\end{theorem}


\begin{proof}
First, we apply Lemma~\ref{lem:avishay} to $F$ for some parameter $t$ to be defined later. We obtain a read-once formula $F'$ of size $k=O(s/t)$, and $k$ formulas $T_1,\ldots,T_k$ each of size $\leq t$. 

Let $p$ be an approximate polynomial (over the reals) for $F'$ of degree $d=O(\sqrt{k})$ with error $1/10$, guaranteed by Theorem~\ref{thm:reichardt}. Applying Theorem~\ref{thm:sherstov}, we get a $1/10$-robust polynomial $p'$ for $p$ of degree $d'=O(\sqrt{k})$. 

By the hypothesis of the theorem, we know that each $T_i$ has a probabilistic polynomial of degree $O(t^{1-\delta})$ with error $\eps=1/3$. For each $T_i$, draw $O(\log{s})$ independent copies of this probabilistic polynomial, and take their majority vote with an $O(\log{s})$-degree polynomial. For an appropriate leading constant in the big-O, we can obtain a probabilistic polynomial for $T_i$ of degree $O(t^{1-\delta}\cdot \log{s})$ with error $1/(10s)$.

Let ${\cal{D}}_1,\ldots,{\cal{D}}_k$ be probabilistic polynomials of degree $D=O(t^{1-\delta}\cdot \log{s})$ with error $\eps=1/(10s)$ for the formulas $T_1,\ldots, T_k$. The error bound $\eps=1/(10s)$ guarantees that for every $x\in\Ftwo^n$, \emph{all} $k$ polynomials compute the correct value with probability at least $9/10$.

Now for every $T_i$, we compute the average $A_i$ (over the reals) of $O(n)$ independent samples from ${\cal{D}}_i$. By a Chernoff bound and union bound,  each $A_i$ is within $\pm1/10$ of the correct 0-1 value for $T_i$, over all $2^n$ inputs $x$, with probability of error $1/\exp(n)$. By the properties of robust polynomials, $p'$ fed the sums $A_i$ will still output the correct value (within $\pm1/10$) for \emph{all} inputs $x\in\Ftwo^n$, for some choice of samples.

Therefore $F$ can be computed by a
\[\SUM_{n^{d'}}\circ\PRODUCT_{d'}\circ \SUM_{O(n)}\circ \MOD2\circ \AND_{D}.\] 
Applying distributivity to the PRODUCT of SUMs, we get 
\[\SUM_{n^{d'}}\circ\SUM_{n^{O(d')}}\circ \PRODUCT_{d'}\circ \MOD2\circ \AND_{D}.\] 
Noting the PRODUCTs now take 0/1 inputs, we can replace them with ANDs: 
\[\SUM_{n^{d'}}\circ\SUM_{n^{O(d')}}\circ \AND_{d'}\circ \MOD2\circ \AND_{D}.\]
Taking the Fourier expansion of the AND function (see, e.g., \eqref{eq:and fourier} in Lemma~\ref{lem:fourier}), we can replace each AND gate with a SUM of $2^{d'}$ MOD2s of fan-in $\leq d'$:
\[\SUM_{n^{d'}}\circ\SUM_{n^{O(d')}}\circ \SUM_{2^{d'}}\circ \MOD2\circ \AND_{D}.\]
Merging the SUMs, our final expression has the form:
\[\SUM_{n^{O(d')}}\circ \MOD2\circ \AND_{D}.\]
Finally, we want to choose a value of $t$ so that the fan-in of the SUM is subexponential, and the fan-ins of the AND's are sublinear (which will also imply that the fan-in of the MOD2's are sub-exponential). Let $t=n^{1+\beta}$, where $\beta$ is an arbitrary number between $\alpha<\beta<\delta/(1-\delta)$. Note that
\[d'=O(\sqrt{k})=O(\sqrt{s/t})
=O(n^{1-\frac{\beta-\alpha}{2}})=O(n^{1-\gamma})
\]
for every $0<\gamma<\frac{\beta-\alpha}{2}$.
Also, observe that 
\[
D=O(t^{1-\delta}\cdot\log{s}) = O(n^{1-(1-\delta)(\delta/(1-\delta)-\beta)}\log{n})=O(n^{1-\gamma})
\]
for every $0<\gamma<(1-\delta)(\delta/(1-\delta)-\beta)$.

From the upper bounds on $d'$ and $D$, we have that $F$ can be computed by \[\SUM_{2^{n^{1-\gamma}}} \circ \MOD2_{2^{n^{1-\gamma}}} \circ \AND_{n^{1-\gamma}} \;
\]
for some $\gamma>0$.
\end{proof}

The above formula depth reduction shows that, if there are more efficient probabilistic polynomials for DeMorgan formulas (and we have no reason to doubt this), then super-cubic formulas have interesting representations as approximate sums of sub-exponentially many sub-linear degree $\F_2$-polynomials. Recent work~\cite{Williams18,Chen-Williams19} can already be applied to prove interesting lower bounds against approximate sums of $2^{n^{\alpha}}$ $\F_2$-polynomials of degree $n^{\beta}$, where $\alpha+\beta < 1$. The remaining challenge will be to prove lower bounds when $\max\{\alpha,\beta\} < 1$.

\section{Circuit Depth Reductions}\label{sec:reductions}
In this section, we present new depth reductions for circuits with unrestricted depth. 

\subsection{Linear Circuits}
We start by considering \emph{linear} circuits, i.e., circuits consisting of $\oplus$ gates only. For technical reasons, we assume that there are $n+1$ input gates in a~linear circuit: $x_1, \dotsc, x_n$ as well as the constant~0. For a~matrix $M \in \{0,1\}^{m \times n}$, we say that a~linear circuit $\mathcal{C}$ with $m$ outputs computes the linear transformation $M$ 
 if the $i$-th output of $\mathcal{C}(x)$ equals the $i$-th row of $Mx$ for all $x \in \{0,1\}^n$, treating $\mathcal{C}(x)$ as the vector of output values. 
 We say that a~linear circuit~$\mathcal{C}$ computing $M$ is \emph{optimal} if no circuit of smaller size computes~$M$.

The main result of this subsection asserts that matrices computable by small linear circuits are not too rigid. The contrapositive says: to get an improved lower bound on the size of linear circuits, it suffices to construct a~matrix with good rigidity parameters. Below, we restate the corresponding theorem formally and then prove it.

\lineartransformation*


\begin{proof}
Let $\mathcal{C}$ be an optimal~circuit of size~$s$ computing~$M$.
If $s < 16$ or the depth of~$\mathcal{C}$ is at most~4, then
each output depends on at most~$16$ variables. Hence $M$~is 16-sparse and the theorem statement holds. Consider this as the base case of an induction on~$s$. 

For the induction step, we ``normalize'' $\mathcal{C}$. Namely, we show how to express $M$ as the (modulo 2) sum of two $\F_2$-matrices $A$ and $B$, where $A$ is $16$-sparse (each row has $\leq 16$ ones) and $B$ has rank at most $\lfloor s/4 \rfloor$. Note that if $\mathcal{C}$~has an~output gate~$H$ of depth at most~4, 
then $H$~depends on at most~$2^4=16$ inputs. Thus the corresponding row $r_H$ of~$M$ has at most~16 ones. Consider the $(m-1)\times n$ matrix~$M_{-H}$ obtained by removing $r_H$ from $M$. We claim that $\R_{M_{-H}}(\floor{s/4}) \leq 16$ implies $\R_{M}(\floor{s/4}) \leq 16$. Indeed, suppose $M_{-H}=A_{-H} \oplus B_{-H}$ where $A_{-H}$ is 16-sparse and $\operatorname{rank}(B_{-H}) \le \lfloor s/4 \rfloor$. To get matrices $A$ and $B$ for~$M$, we simply add the row $r_H$ to $A_{-H}$ and a corresponding all-zero row to~$B_{-H}$. Clearly, the resulting matrix~$A$ is 16-sparse and the rank of the resulting matrix~$B$ does not change. Thus, in the following, we assume WLOG that $\mathcal{C}$~has no output gates of depth at~most~$4$. Our crucial step is the following claim.

\begin{claim}
\label{claim:lincaseanalysis}
Let $\mathcal{C}$ be an optimal linear circuit computing $M \in \{0,1\}^{m \times n}$ such that $s(\mathcal{C}) \ge 16$, and no output gate of~$\mathcal{C}$ has depth smaller than~5. Then there is a gate $G$ in $\mathcal{C}$ and a~linear circuit~$\mathcal{C}'$ computing a matrix $M' \in \{0,1\}^{m \times n}$ with the properties:
\begin{enumerate}
\item $s(\mathcal{C}') \le s(\mathcal{C})-4$, and
\item for every $x \in \{0,1\}^n$, if $G(x)=0$ then $\mathcal{C}(x)=\mathcal{C'}(x)$.
\end{enumerate}
\end{claim}

For now, suppose the claim is proved. Consider the circuit~$\mathcal{C}'$, gate $G$ in $\mathcal{C}$, and matrix~$M'$ provided by Claim~\ref{claim:lincaseanalysis}. Let $g \in \{0,1\}^{1 \times n}$ be the~characteristic vector of the linear function computed by $G$, so that $G(x)=gx$. By the claim, $gx=0$ implies $(M \oplus M')x=0$. Hence $(M \oplus M')$ is either the zero matrix, or it defines the same linear subspace as~$g$: $M \oplus M'=tg$ for a~vector $t \in \{0,1\}^{m \times 1}$.

By the induction hypothesis, $M'=A' \oplus B'$ where $A'$ is 16-sparse, and $\operatorname{rank}(B') \le \floor{\frac{s-4}{4}} = \floor{\frac s4}-1$. Thus, $M=A' \oplus B$, where the matrix $B=B' \oplus tg$ has rank at most $\floor{s/4}$ by subadditivity of the rank function. 
\end{proof}

We now turn to proving the remaining claim.

\begin{proof}[Proof of Claim~\ref{claim:lincaseanalysis}]
\leavevmode
\begin{description}
\item[Case 1:] {\bf There is a~gate $G$ in~{\C} of depth at least~2 and at most~4, and has out-degree at least~$2$.} Let the predecessors of $G$ be~$B$ and~$C$, and call two of its successors~$D$ and~$E$, see Figure~\ref{fig:transformationsclaim}
(in this and the following figures, we write the out-degrees of some of the gates near them).
The circuit~$\mathcal{C}'$ is obtained from~$\mathcal{C}$ by ``assigning'' the output of $G$ to be $0$. Note that $B(x)=C(x)$ for all $x \in \{0,1\}^n$ where $G(x)=0$. At least one of~$B$ and~$C$ must be an internal gate (otherwise $G$~would have depth~1), let it be~$C$. Since~$C$ computes the same function as~$B$, it may be removed from~$\mathcal{C}'$: we remove it, and replace every wire of the form $C \to H$ by a~new wire $B \to H$. Note that neither~$G$ nor~$C$ is an output gate. Now, we show that both~$D$ and~$E$ can also be removed. Let us focus on the gate~$D$ (for~$E$ it is shown similarly) and call its other predecessor~$F$. Since~$G=0$, the gate~$D$ computes the same function as~$F$. This means that one may remove~$D$: we remove it and replace every wire~$D \to H$ by a~wire~$F \to H$. If~$D$~happens to be an~output gate, we move the corresponding output label from~$D$ to~$F$.

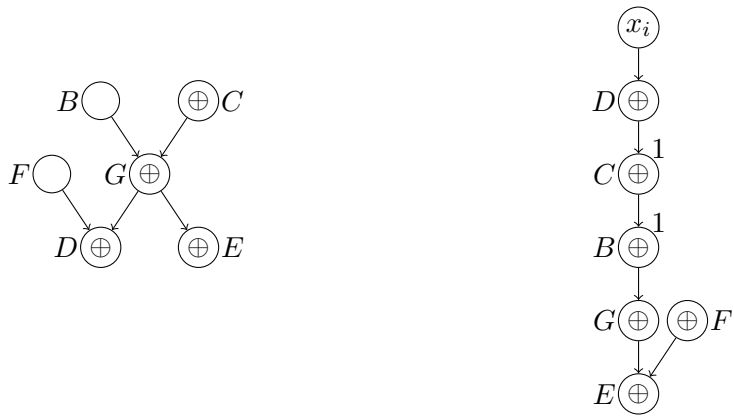
\begin{figure}[ht]
\label{fig:transformationsclaim}
\begin{mypic}
\tikzstyle{d}=[below,draw=none,rectangle,inner sep=0mm,text width=55mm]

\begin{scope}
\node[label=180:$B$] (b) at (0,5.5) {};
\node[label=0:$C$] (c) at (1,5.5) {$\oplus$};
\node[label=180:$G$] (g) at (0.5,4.75) {$\oplus$};
\node[label=180:$D$] (d) at (0,4) {$\oplus$};
\node[label=0:$E$] (e) at (1,4) {$\oplus$};
\node[label=180:$F$] (f) at (-0.5,4.75) {};
\foreach \f/\t in {b/g, c/g, g/d, g/e, f/d} \draw[->] (\f) -- (\t);
\node[d] at (0.5,1.75) {Case~1: assuming $G=0$, the gate~$G$ is removed, $B$~is replaced by~$C$, and $D$~and~$E$ are replaced by their other predecessors.};
\end{scope}

\begin{scope}[xshift=50mm]
\node (x) at (0.5,6.25) {$x_i$};
\node[label=180:$D$] (d) at (0.5,5.5) {$\oplus$};
\node[label=60:1,label=180:$C$] (c) at (0.5,4.75) {$\oplus$};
\node[label=60:1,label=180:$B$] (b) at (0.5,4) {$\oplus$};
\node[label=180:$G$] (g) at (0.5,3.25) {$\oplus$};
\node[label=180:$E$] (e) at (0.5,2.5) {$\oplus$};
\node[label=0:$F$] (f) at (1,3.25) {$\oplus$};
\foreach \f/\t in {x/d, d/c, c/b, b/g, g/e, f/e} \draw[->] (\f) -- (\t);
\node[d] at (0.5,1.75) {Case~2: assuming $G=0$, the gates $B$, $C$, and $G$ are removed whereas $E$ is replaced by~$F$.};
\end{scope}

\end{mypic}
\caption{Cases in the proof of Claim~\ref{claim:lincaseanalysis}.}
\end{figure}

\item[Case 2:] {\bf All gates of depth at least~2 and at most~4 have out-degree exactly~1 in~\C.} Take a~gate~$G$ of depth~$4$ and trace back its longest path to an input: $x_i \to D \to C \to B \to G$. Let also $E$~be the successor of~$G$ (which exists because $C$ has depth at least $5$). By assumption, gates~$B$ and~$C$ have out-degree~1. This means that in~$\mathcal{C}$ they are only used for computing the gate~$G$. This, in turn, means that assuming $G=0$, we can remove~$G$,~$B$, and~$C$ (note none of them is an output). Finally, the gate~$E$ can be replaced by the other input $F$ of~$E$ (note $F \notin \{B,C,G\}$, since $\mathcal{C}$ is optimal).
\end{description}
This completes the proof.\end{proof}

\begin{remark}
Extending the same ideas, one can show
that any linear circuit $\C$ of size~$s$ can be computed by 
an~\dm{2^{\ceil{\frac s4}}}{s\cdot 2^{14}}{16} circuit. For this, one considers two optimal circuits~$\C_0$ and~$\C_1$ resulting from~$\C$ by assuming~$G=0$ and~$G=1$, respectively. As shown in the proof, both $\C_0$ and $\C_1$ have size at most $s-4$. One then proceeds by induction. We illustrate this approach in full detail in the next subsection. 
\end{remark}

\begin{remark}
The proof of Theorem~\ref{thm:linear} gives a~decomposition $M=A \oplus B=A\oplus (C\cdot D)$, where $A\in\F^{m\times n}$ is $16$-sparse, $C\in\F^{m\times s/4}$ is composed of 
vectors~$t$, and $D\in\F^{s/4 \times n}$ is composed of vectors~$g$. Since the chosen gate~$G$ always has depth at most four, the vector $g$ is $16$-sparse. Thus, we in fact have a decomposition $M=A\oplus(C\cdot D)$, where both $A$~and~$D$ are $16$-sparse. In particular, the row-space of $M$~is spanned by the union of row-spaces of~$A$ and~$D$. This implies that the row-space of $M$ can be spanned by at most $(m+\frac s4)$ $16$-sparse vectors. The corresponding matrix property is called \emph{outer dimension}, and it is studied in~\cite{PP06,Lokam09}. While the current lower bounds on the outer dimension of explicit matrices do not lead to new circuit lower bounds, it would be interesting to study their applications in this context.
\end{remark}

\subsection{General Boolean Circuits}
In this section, we study the following natural question: given a~Boolean circuit\footnote{In this section we consider functions with one output, but these results can be trivially generalized to the multi-output case.} and given an integer $k \geq 2$, what is the smallest \dz{k} circuit computing the same function? 
To this end, we introduce the following notation. For an integer $k \ge 2$, we define $\alpha(k)$ as the infimum of all values~$\alpha$ such that any circuit of size~$s$ can be rewritten as a~$\dt{2^{\alpha s}}{k}$ circuit. 

\newcommand{\df}[3]{\ensuremath{\text{OR}_{#1} \circ \text{AND}_{#2} \circ C(#3)}}

For proving upper bounds on $\alpha(k)$ it will be convenient to consider the following class of circuits. Let \df{p}{q}{r} be a~class of circuits with an output OR that is fed by at most~$p$ AND's of at most~$q$ circuits of size at most~$r$.

\begin{theorem}\label{thm:transformations}
Every~circuit of size~$s$ can be computed as:
\begin{enumerate}
\item an \df{2^{\ceil{\frac s2}}}{\ceil{\frac s2}}{1} circuit;
\item an \df{2^{\ceil{\frac s{3.9}}}}{\ceil{\frac s3}}{15} circuit.
\end{enumerate}
\end{theorem}

Note that any circuit of size~$r$ depends on at most $r+1$ variables, and hence can be written as an $(r+1)$-CNF with at most~$2^{r}$ clauses. Therefore every $\df{p}{q}{r}$ circuit can be easily converted into a~\dm{p}{q2^{r}}{r+1} circuit. 
Thorem~\ref{thm:maintransformation}, which we restate below, is then an immediate corollary of Theorem~\ref{thm:transformations}. In turn, it implies that $\alpha(2) \le \frac 12$ and $\alpha(16) \le \frac 1{3.9}$.

\maintransformation*

\begin{proof}[Proof of Theorem~\ref{thm:transformations}]
Both parts are proven in a~similar fashion. We proceed by induction on~$s$. The base case is when $s$~is small. We then just have an~\df{1}{1}{s} circuit.

For the induction step we take a~gate~$G$ of~$\mathcal{C}$ and consider two circuits~$\mathcal{C}_0$ and~$\mathcal{C}_1$ where $\mathcal{C}_i$~computes the same as $\mathcal{C}$ on all inputs $\{x \in \{0,1\}^n \colon G(x) = i\}$. We may assume both $\mathcal{C}_i$'s are minimal size among all such circuits. Since $\mathcal{C}_i$~can be obtained from~$\mathcal{C}$ by removing the gate~$G$ (as it computes the constant~$i$ on the corresponding subset of the Boolean hypercube), we conclude that $s(\mathcal{C}_i) < s$. This allows us to proceed by induction. Assume that by the induction hypothesis $\mathcal{C}_i$ is guaranteed to be expressible as an
$\df{p_i}{q_i}{r_i}$ circuit. We use the following identity to convert~$\mathcal{C}$ into the required circuit:
\begin{equation}
\label{eq:ind}
\mathcal{C}(x) \equiv ([G(x)=0] \land \mathcal{C}_0(x)) \lor ([G(x)=1] \land \mathcal{C}_1(x)) \, .
\end{equation}
Assume that the subcircuit of~$\mathcal{C}$ computing the gate~$G$ has at most~$t$ gates. We claim that $[G(x)=i] \land \mathcal{C}_i$ can be written as an~$\df{p_i}{q_i+1}{\max\{r_i, t\}}$ circuit. For this, we just feed a~new circuit computing~$G$ to every AND gate. Plugging this into~\eqref{eq:ind}, gives an
\begin{equation}
\label{eq:ind2}
\df{p_0+p_1}{\max\{q_0,q_1\}+1}{\max\{t, r_0, r_1\}}
\end{equation}
circuit for computing $\mathcal{C}$.

Below, we provide details specific to each of the two items from the theorem statement. In particular, we estimate the parameters $p_i$'s, $q_i$'s, $r_i$'s, and~$t$ and plug them into~\eqref{eq:ind2}.

\begin{enumerate}
\item The base case is $s=1$. Then $\mathcal{C}$ consists of a~single gate and can be expressed as an~\df{1}{1}{1} circuit.
For the induction step, assume that $s \ge 2$ and take a~gate~$A$ that depends on two variables. Let $G=A$, hence $t=1$. The gate~$A$ must have at least one successor (otherwise $\mathcal{C}$ can be replaced by a~circuit with smaller than $s$~gates). Clearly, $A$~and its successors are not needed in $\mathcal{C}_i$'s. Hence, by the induction hypothesis $p_i \le 2^{\frac{s-2}{2}+1}$, $q_i \le \frac{s-2}{2}+1$, $r_i \le 1$. Plugging this into~\eqref{eq:ind2} gives the desired result.

\item Take a~gate~$A$ that is fed by two variables $x$~and~$z$ and has the maximum distance to an output. If its distance to output is at most~$4$, then $s(\mathcal{C}) \le 15$ and we just rewrite it as an~\df{1}{1}{15} circuit. This is the base case. Assume now that the distance from~$A$ to the output gate is at least~5. In the analysis below, we always ``follow'' the longest path from~$A$ to the output. This allows us to conclude that any such path is long enough and hence each gate considered has positive out-degree (i.e., is not an output). Moreover, each gate on this path cannot depend on too many variables. Let $B$~be a~successor of~$A$ on the longest path to the output.

In the five cases below, we show that we can always find a~gate~$G$ that $s(G) \le 15$ and both $s(\mathcal{C}_0)$ and $s(\mathcal{C}_1)$ are small enough. In particular, 
$s(\mathcal{C}_0), s(\mathcal{C}_1) \le s-4$ works for us: 
$p_0+p_1 \le 2\cdot 2^{\ceil{\frac{s-4}{3.9}}} < 2^{\ceil{\frac{s}{3.9}}}$, $\max\{q_0,q_1\}+1 \le \ceil{\frac{s-4}{3}}+1 < \ceil{\frac s3}$. 

See Figure~\ref{fig:transformations} for an illustration of the five cases. For a~gate~$G$, by $\out(G)$ we denote the out-degree of~$G$.

\begin{figure}[ht]
\label{fig:transformations}
\begin{mypic}
\tikzstyle{d}=[below,draw=none,rectangle,inner sep=0mm,text width=30mm]

\begin{scope}
\node (x) at (0,6) {$x$};
\node (y) at (1,6) {$z$};
\node[label=180:$A$] (a) at (0.5,5.5) {};
\node[label=60:1,label=180:$B$] (b) at (0.5,4.75) {};
\node[label=60:1,label=180:$C$] (c) at (0.5,4) {};
\node[label=180:$E$] (e) at (0.5,3.25) {};
\foreach \f/\t in {x/a, y/a, a/b, b/c, c/e} \draw[->] (\f) -- (\t);
\node[d] at (0.5,2.75) {Case 1.1: when $E$ is constant, one removes~$B$, $C$, $E$, and successors of~$E$.};
\end{scope}

\begin{scope}[xshift=24mm]
\node (x) at (0,6) {$x$};
\node (y) at (1,6) {$z$};
\node[label=180:$A$] (a) at (0.5,5.5) {};
\node[label=60:1,label=180:$B$] (b) at (0.5,4.75) {};
\node[label=60:$2^+$,label=180:$C$] (c) at (0.5,4) {};
\node (e) at (0,3.25) {};
\node (f) at  (1,3.25) {};
\foreach \f/\t in {x/a, y/a, a/b, b/c, c/e, c/f} \draw[->] (\f) -- (\t);
\node[d] at (0.5,2.75) {Case 1.2: when $C$~is constant, one removes~$B$, $C$, and successors of~$C$.};
\end{scope}

\begin{scope}[xshift=48mm]
\node (x) at (0,6) {$x$};
\node (y) at (1,6) {$z$};
\node[label=180:$A$] (a) at (0.5,5.5) {};
\node[label=180:$D$] (d) at (-0.15,5.25) {};
\node[label=60:$2^+$,label=180:$B$] (b) at (0.5,4.75) {$\oplus$};
\node (e) at (0,4) {};
\node (f) at  (1,4) {};
\foreach \f/\t in {x/a, y/a, a/b, b/e, b/f, d/b} \draw[->] (\f) -- (\t);
\node[d] at (0.5,2.75) {Case 2.1: when $B$~is constant, one removes~$B$ and its successors, replace~$A$ by~$D \oplus c$.};
\end{scope}

\begin{scope}[xshift=72mm]
\node (x) at (0,6) {$x$};
\node (y) at (1,6) {$z$};
\node[label=180:$A$,label=90:$1$,] (a) at (0.5,5.5) {};
\node[label=60:$2^+$,label=180:$B$] (b) at (0.5,4.75) {$\land$};
\node (e) at (0,4) {};
\node (f) at  (1,4) {};
\foreach \f/\t in {x/a, y/a, a/b, b/e, b/f} \draw[->] (\f) -- (\t);
\node[d] at (0.5,2.75) {Case 2.2.1: when $B$~is constant, one removes~$B$ and its successors, and~$A$.};
\end{scope}

\begin{scope}[xshift=96mm]
\node (x) at (0,6) {$x$};
\node (y) at (1,6) {$z$};
\node[label=180:$A$,label=90:$2^+$,] (a) at (0.5,5.5) {};
\node[label=60:$2^+$,label=180:$B$] (b) at (0.5,4.75) {$\land$};
\node (e) at (0,4) {};
\node (f) at  (1,4) {};
\foreach \f/\t in {x/a, y/a, a/b, b/e, b/f} \draw[->] (\f) -- (\t);
\node[d] at (0.5,2.75) {Case 2.2.2: when~$B$ is constant, one removes~$B$ and its successors; moreover, $B=1$ it forces~$A$ to be a~constant and removes~$A$ and its successors.};
\end{scope}
\end{mypic}
\caption{Cases in the proof of the second part of Theorem~\ref{thm:transformations}.}
\end{figure}

\begin{description}
\item[Case 1:] $\out(B)=1$. Let $C$~be the successor of~$B$.
\begin{description}
\item[Case 1.1:] $\out(C)=1$. Let $E$~be the successor of~$C$. Let $G=E$. In $\mathcal{C}_i$'s, one removes~$B$, $C$ (as they were only needed to compute~$E$ that is now a~constant), $E$, and the successors of~$E$. 
\item[Case 1.2:] $\out(C)\ge 2$. Let $G=C$. In $\mathcal{C}_i$'s, one removes $B$, $C$, and the successors of~$C$.
\end{description}
\item[Case 2:] $\out(B) \ge 2$. Let $D$~be the other input of~$B$. It may be a~gate or an~input variable. If~$B$ computes a~constant Boolean binary operation or an operation that depends on~$A$ or~$D$ only, then $\C$ is not optimal. Otherwise, $B$~computes one of the following two types of functions (either linear or quadratic polynomial over~$\mathbb{F}_2$):
\begin{description}
\item[Case 2.1:] $B(A,D)=A \oplus D \oplus c$ where $c \in \{0,1\}$.  Let $G=B$. In $\mathcal{C}_i$'s, one immediately removes~$B$ and its successors. Also, in $\mathcal{C}_i$, $D \oplus A=i \oplus c$. Hence, $A$~may be replaced by~$D \oplus i \oplus c$.
\item[Case 2.2:] $B(A,D)=(A \oplus a) \cdot (D \oplus d) \oplus c$ where $a, d, c \in \{0,1\}$.
\begin{description}
\item[Case 2.2.1:] $\out(A)=1$. Let $G=C$. In $\mathcal{C}_i$'s, one removes~$B$, its successors, and~$A$.
\item[Case 2.2.2:] $\out(A) \ge 2$. Let $D$~be the other successor of~$B$. 
Let $G=B$. In $\mathcal{C}_i$'s, one removes~$B$ and its successors. Also,  $B=c \oplus 1$ forces $A=a \oplus 1$ and $D=d \oplus 1$. Hence, in $\mathcal{C}_{c \oplus 1}$ two additional gates are removed: $A$ and its successors (if a~successor of~$B$ happens to be a~successor of~$A$ also, then it is a~function on~$A$ and~$D$ and the circuit can be simplified, which contradicts its optimality). Hence, 
\(p_0+p_1 \le 2^{\ceil{\frac{s-3}{3.9}}}+2^{\ceil{\frac{s-5}{3.9}}} \,.\)
This is smaller than $2^{\ceil{\frac{s}{3.9}}}$ since $2^{-\frac{3}{3.9}}+2^{-\frac{5}{3.9}}<1$.
\end{description}
\end{description}
\end{description}

\end{enumerate}
This completes the proof.\end{proof}

\begin{remark}
It is not difficult to see that the output OR gate is a ``disjoint OR'', and can be replaced by a~SUM gate over the integers. In other words, for every $x \in \{0,1\}^n$, at most one subcircuit feeding into the OR gate may evaluate to~1. This holds because we always consider two mutually exclusive cases: $G=0$ or~$G=1$.
\end{remark}

\subsection{Properties of $\alpha(k)$}

We start by observing a~lower bound on $\alpha(k)$.

\begin{lemma}
For any integer $k \ge 2$, $\alpha(k) \ge 1/k$.
\end{lemma}
\begin{proof}
Let $\oplus_n$ denote the parity function of $n$ inputs. It has $2^{n-1}$ inputs where it is equal to~1 and all these inputs are isolated, that is, the Hamming distance between any pair of them is at least $2$. As proven by Paturi, Pudl{\'a}k, and Zane~\cite{PPZ97}, every $k$-CNF has at most $2^{n(1-1/k)}$ isolated satisfying assignments. This implies that $\oplus_n$~cannot be computed by an OR of fewer than $2^{n/k-1}$ $k$-CNFs. Since $s(\oplus_n)=n-1$, this implies that 
\[\alpha(k) \ge \frac{\frac nk - 1}{n-1} \, .\]
Since this must hold for arbitrary large~$n$, $\alpha(k) \ge 1/k$.
\end{proof}

Thus, we know the exact value of $\alpha(2)=\frac{1}{2}$. This immediately implies a circuit lower bound of $2n-o(n)$ for BCH codes. Indeed, it was shown in~\cite{PSZ97} that when the bottom fan-in is restricted to $k=2$, then BCH codes require depth-$3$ circuits of size $2^{n-o(n)}$. And, since $\alpha(2)=\frac 12$, they must have circuit complexity at least $2n-o(n)$.

One can use techniques from Theorem~\ref{thm:transformations} to prove an upper bound of $\alpha(3) \le \frac{\log_2 3}{4}$. Thus, we know that
\[\frac 13 \le \alpha(3) \le \frac{\log_2 3}{4} < 0.3963 \, .\]

We conjecture that the upper bound on $\alpha_3$ is tight. One way to prove this would be to find the $s_3^3$ complexity of the inner product function: $\operatorname{IP}(x_1,\dotsc,x_n)=x_1x_2 \oplus x_3x_4 \oplus \dotsb \oplus x_{n-1}x_n$. In particular, if the upper bound shown in the next lemma is tight, then $\alpha(3)=\frac{\log_2 3}{4}$.

\begin{lemma}
\leavevmode
\begin{enumerate}
\item $2^{\frac n4} \le s_3^2(\IP) \le 2^{\frac n2-o(n)}$.
\item $2^{\frac n6} \le s_3^3(\IP) \le 3^{\frac n4}$.
\end{enumerate}
\end{lemma}
\begin{proof}
Note that by substituting every other input of ${\IP}$ by~1, one gets the parity function $\oplus_{\frac n2}$ on the remaining $n/2$ inputs.
Now both lower bounds follow from the corresponding lower bounds for the parity function: $s_3^2(\oplus_k) \ge 2^{\frac k2}$ and $s_3^3(\oplus_k) \ge 2^{\frac k3}$.
\begin{enumerate}
\item 
The first upper bound follows from the fact that $\IP(x_1, \dotsc, x_n)=1$ iff there is an odd number of ones among
\[p_1=x_1x_2, \, p_2=x_3x_4, \, \dotsc, p_{\frac n2}=x_{n-1}x_n \,.\]
Hence,
\[\IP(x_1, \dotsc, x_n) \equiv \bigvee_{S \subseteq [\frac n2] \colon |S| \bmod 2 = 1}\left(\bigwedge_{i \in S}[p_i=1] \land \bigwedge_{i \not \in S}[p_i=0]\right) \, .\]
It remains to note that each $[p_i=c]$ can be expressed as a~2-CNF because $p_i$ depends on two variables.

\item 
For the second upper bound, note that $\IP(x_1, \dotsc, x_n)=1$ iff there is an odd number of 1's among
\[p_1=x_1x_2\oplus x_3x_4, \, p_2=x_5x_6\oplus x_7x_8, \, \dotsc, p_{\frac n4}=x_{n-3}x_{n-2}\oplus x_{n-1}x_n \,.\]
To compute $\IP$ by a~depth~3 circuit, we go through all possible $2^{\frac n4 - 1}$ values of $p_1, \dotsc, p_{\frac n4}$ such that an odd number of them is equal to~1:
\begin{equation}
\label{eq:ipcnf}
\IP(x_1, \dotsc, x_n) \equiv \bigvee_{S \subseteq [\frac n4] \colon |S| \bmod 2 = 1}\left(\bigwedge_{i \in S}[p_i=1] \land \bigwedge_{i \not \in S}[p_i=0]\right)
\end{equation}
Now, we show that $[p_i=0]$ can be written as a~single 3-CNF, whereas $[p_i=1]$ can be expressed as an OR of two 3-CNFs. W.l.o.g. assume that $i=1$. The clauses of a~3-CNF expressing $[p_i=0]$ should reject all assignments to $x_1,x_2,x_3,x_4 \in \{0,1\}$ where $\IP(x_1,x_2,x_3,x_4)=1$. In all such assignments, one of the two monomials ($x_1x_2$ and $x_3x_4$) is equal to~0 whereas the other one is equal to~1. Hence, one needs to write down a~set of clauses rejecting the following four partial assignments: $\{x_1=0, x_3=x_4=1\}$, $\{x_2=0, x_3=x_4=1\}$, $\{x_1=x_2=1, x_3=0\}$, $\{x_1=x_2=1, x_4=0\}$. Thus, 
\[
[p_1(x_1,x_2,x_3,x_4)=0] \equiv 
(x_1\lor \neg x_3 \lor \neg x_4)
\land
(x_2\lor \neg x_3 \lor \neg x_4)
\land
(\neg x_1 \lor \neg x_2 \lor x_3)
\land
(\neg x_1 \lor \neg x_2 \lor x_4) \, .
\]
In turn, to express $[p_1=1]$ as an OR of two 3-CNFs we consider both assignments to~$x_1$:
\[[p_1(x_1,x_2,x_3,x_4)=1] \equiv \left((x_1) \land [x_2 \oplus x_3x_4 = 0]\right) \lor ((\neg x_1) \land [x_3x_4=1]) \, .\]
It remains to note that each of $[x_2 \oplus x_3x_4=0]$ and $[x_3x_4=1]$ can be written as a~3-CNF. Let $[p_i=0] \equiv P_i$ and $[p_i=1] \equiv ((x_i) \land Q_i) \lor ((\neg x_i) \land R_i)$ where $P_i$, $Q_i$, and $R_i$ are 3-CNFs. One may then expand~\eqref{eq:ipcnf} as follows:
\[
 \bigvee_{S \subseteq [\frac n4] \colon |S| \bmod 2 = 1}
 \left(
 \bigvee_{T \subseteq S} 
 \left(
 \bigwedge_{i \in T}\left((x_i) \land Q_i\right) \land \bigwedge_{i \in S \setminus T}((\neg x_i) \land R_i) \land \bigwedge_{i \not \in S}P_i
 \right)
 \right)
\]

The fan-in of the resulting OR-gate is
\[\sum_{S \subseteq [\frac n4] \colon |S| \bmod 2 = 1}2^{|S|} \le \sum_{i=0}^{\frac n4}\binom{n/4}{i}2^i =3^{\frac n4} \, .\]
\end{enumerate}
\end{proof}

\begin{openproblem}
Determine $s_3^3(\operatorname{IP})$.
\end{openproblem}

Besides finding the exact values of $\alpha(k)$, it would be interesting to find out whether every circuit of \emph{linear size} can be computed by a~non-trivial depth~3 circuit with constant bottom fan-in. We restate this open problem below.

\ouropenproblem*

This paper supports the conjecture by showing that it holds for small values of $c$. As another example, we can consider a class of functions where we know linear \emph{upper bounds} on circuit complexity.
For any \emph{symmetric} function~$f$ (i.e., a function whose value depends only on the sum over integers of the input bits) we know that $s(f) \le 4.5n+o(n)$~\cite{DKKY10}. It is also known~\cite{PSZ97,wolf} that symmetric functions can be computed by relatively small depth-$3$ circuits: $s_3^k(f)\le\operatorname{poly}(n) \cdot \left(1+1/k\right)^n$ (and this bound is tight~\cite{wolf}).


Since in our depth reduction results, we always get $k$-CNFs with small linear number of clauses, 
it is interesting to study the expressiveness of OR of exponential number of such $k$-CNFs. Let us define $\alpha(k,c)$ as the infimum of all values $\alpha$ such that any circuit of size at most $cn$ can be computed as an $\dm{2^{\alpha n}}{cn}{k}$. We can upper bound the rate of convergence of $\alpha(k,c)$ using the following width reduction result for CNF-formulas~\cite{schuler2005algorithm,calabro2006duality}.
\begin{theorem}[\cite{schuler2005algorithm,calabro2006duality}]\label{thm:width}
For any constant $0<\eps\leq1$ and a function $C\colon\mathbb{N}\to\mathbb{N}$, any CNF formula $f$ with $n$ variables and $n\cdot C(n)$ clauses can be expressed as $f=\text{OR}_{i=1}^{t} f_i$, where $t\leq 2^{\eps n}$ and each $f_i$ is a $k$-CNF formula with at most $n\cdot C(n)$ clauses, where $k=O\left(\frac{1}{\eps}\cdot\log\left(\frac{C(n)}{\eps}\right)\right)$.
\end{theorem}

For our applications, we are interested in $\alpha(k,c)$ for small fixed $c$.
Since for every $c$, $\alpha(k,c)$ is a non-increasing bounded sequence, we let $\alpha(\infty, c)=\lim_{k\to\infty}\alpha(k,c)$. Then Theorem~\ref{thm:width} implies that $\alpha(k,c)\geq \alpha(\infty,c)\geq \alpha(k,c)-O\big(\frac{\log(ck)}{k}\big)$.

\section{Applications}\label{sec:app}
In this section, we state formally the results that are presented in the last three row-blocks of~Table~\ref{table:comp}. Namely, we show that improving the parameters for the known explicit constructions
of the following pseudorandom objects imply circuits lower bounds via depth reduction techniques presented in the previous section:
\begin{itemize}
\item functions that are not constant on any large algebraic variety in $\{0,1\}^n$ defined by polynomials of small degree (such functions are called dispersers);
\item functions that agree with any polynomial of small degree on roughly half of the points in $\{0,1\}^n$;
\item matrices that are far from matrices of small rank.
\end{itemize} 
For comparison, we also show what these tools give when applied to Valiant's reductions.


\subsection{Dispersers}
In this section we show that dispersers for algebraic varieties over $\mathbb{F}_2$ cannot be computed by small circuits. We note that dispersers for varieties of degree one have been used for proving lower bounds on unrestricted circuits~\cite{DK11,FGHK16}, and it is known that an explicit construction of a disperser for varieties of degree two would slightly improve the known circuit lower bounds~\cite{GK16}. Now we show that dispersers for varieties of degree $16$ will give new circuit lower bounds via a new simple method.

\begin{definition}
A~set $S \subseteq \Ftwo^n$ is called an \emph{$(d,m)$-variety} if it is a set of common roots of at most $m$~polynomials of degree at most~$d$:
\[S = \{x \in \Ftwo^n \colon p_1(x)=\dotsb=p_m(x)=0, \, \deg(p_i) \le d \text{ for all $1 \le i \le m$}\} \, .\]
A set~$S$ is called a \emph{$d$-variety} (or a~variety of degree~$d$) if it is an $(d,\infty)$-variety.
\end{definition}

\begin{definition}
A~Boolean function $f \colon \Ftwo^n\rightarrow\Ftwo$ is called a {\em $(d,m,s)$-disperser} (for parameters $d, m$, and $s$ which possibly depend on $n$) if $f$~is non-constant on any $(d,m)$-variety $S \subseteq \Ftwo^n$ of size larger than~$s$.
\end{definition}

We will make use of the Sparsification Lemma first proven by Impagliazzo, Paturi and Zane~\cite{IPZ}. The dependence of $C$ on $k$ was later improved in~\cite{calabro2006duality}. (And this is essentially tight by~\cite{miltersen2005converting}.)
\begin{theorem}[Corollary 1 in~\cite{IPZ}, Section~6 in~\cite{calabro2006duality}]\label{thm:sparsification}
For all $\eps>0$ and positive $k$, there exists $C$ such that any $k$-CNF formula $f$ with $n$ variables can be expressed as $f=\text{OR}_{i=1}^{t} f_i$, where $t\leq 2^{\eps n}$ and each $f_i$ is a $k$-CNF formula with at most $Cn$ clauses, where $C=O\left(\left(\frac{k}{\eps}\right)^{3k}\right)$.
\end{theorem}

Now we are ready to state the main result of this section.
\begin{theorem}\label{thm:dispersers}
Let $f\colon\F^n\to\F$ be a~function with $|f^{-1}(1)|\geq|f^{-1}(0)|$ and $\eps>0$ be a~constant.\footnote{If $|f^{-1}(1)|<|f^{-1}(0)|$, one can consider the negation of $f$, since taking negations does not change the disperser parameters.}
\begin{itemize}
\item If $f$ is an $(16, 1.3(1-\eps)n, 2^{\eps n})$-disperser, then
$s(f) \ge 3.9(1-\eps)n-4$.
\item If $f$ is an $(\omega(1), O(n), 2^{(1-\eps)n})$-disperser, 
then 
$s_{\text{sp}}(f)=\omega(n)$.
\item If $f$ is $(2^{(\log{n})^{1-o(1)}}, \infty, 2^{(1-\eps) n})$-disperser,
then 
$s_{\log}(f)=\omega(n)$.
\item If $f$ is $(n^{\eps}, \infty, 2^{n-\omega(n/\log\log{n})})$-disperser,
 then 
$s_{\log}(f)=\omega(n)$.
\end{itemize}
\end{theorem}
\begin{proof}
\leavevmode
\begin{itemize}
\item From Theorem~\ref{thm:transformations}, we know that if $f$ is computable by a circuit of size~$s$, then $f$ is also computable by a circuit $\C \in \ensuremath{\text{OR}_{2^{s/3.9}} \circ \text{AND}_{s/3} \circ C(15)}$. Let $t=2^{s/3.9}$, and let $f_1,\ldots,f_t\colon\F^n\to\F$ be the $t$~functions computed in the gates of the $\text{AND}$ level of $\C$. Since $f=\text{OR}_{i=1}^{t} f_i$, we have that
$f^{-1}(1)=\bigcup_{i=1}^{t} f_i^{-1}(1)$.
Thus,
\begin{align}\label{eq:size}
2^{n-1}\leq\left|f^{-1}(1) \right| \leq \sum_{i=1}^{t} |f_i^{-1}(1)| \leq t\cdot\max_{i} |f_i^{-1}(1)| \, .
\end{align}
Each $f_i$ is an $\ensuremath{\text{AND}_{s/3} \circ C(15)}$, that is, a~set of common roots of $s/3$~polynomials of degree $16$ (recall that over $\mathbb{F}_2$ every monomial is multilinear; hence a~circuit of size~15 computes a~polynomial of degree at most~16). Since $f$~is a~disperser for varieties of size $2^{\eps n}$ defined by~$s/3$ polynomials of degree~$16$, each $f_i^{-1}(1)\leq 2^{\eps n}$. Now, \eqref{eq:size}~implies that $s/3.9\geq n-\eps n -1$.
\item The proofs of items (2)--(4) of this theorem follow the same pattern, so we only present the proof of the second item. Assume, towards a contradiction, that an $(\omega(1), O(n), 2^{(1-\eps)n})$-disperser $f$ can be computed by a series-parallel circuit of size $cn$. From Theorem~\ref{thm:valiant}, such a circuit can be expressed as a~circuit $\C \in \dt{2^{\frac{\eps n}{3}}}{k}$ for $k=k(c, \eps)$. By Theorem~\ref{thm:sparsification}, each $k$-CNF computed by the AND gates of $\C$, can be replaced by an $\text{OR}$ of $2^{\frac{\eps n}{3}}$ $k$-CNFs with $Cn$ clauses each where $C=C(\delta,\eps)$. Let $t=2^{\frac{2\eps n}{3}}$, and let $f_1,\ldots,f_t\colon\F^n\to\F$ be the $t$ $k$-CNFs with $Cn$ clauses whose OR computes $f$. Now we have that
each $f_i$ is an $\ensuremath{\text{AND}_{Cn} \circ \text{OR}_{k}}$, that is, a set of common roots of $Cn$ polynomials of degree $k$ (each computing an $\ensuremath{\text{OR}_{k}})$. From the disperser property of $f$, we have that each $f_i$ computes at most $2^{(1-\eps)n}$ ones of $f$. Therefore, in order to compute all $\geq2^{n-1}$ ones of $f$, $t$ must be greater than $2^{\eps n -1}$, which contradicts the definition $t=2^{\frac{2\eps n}{3}}$.
\end{itemize}

\end{proof}

We remark that in the first item of Theorem~\ref{thm:dispersers}, even dispersers for varieties defined by $1.3(1-\eps)n$ functions of $16$ variables (rather than all polynomials of degree $16$) will suffice for proving a lower bound.

In order to prove a new circuit lower bound against unrestricted circuits, it suffices to construct a $(16, 1.05n, 2^{0.2 n})$-disperser. There are known constructions of dispersers for constant-degree varieties over large fields~\cite{D12,ben2012extractors,li2018improved}. For $\mathbb{F}_2$, a long line of work achieved almost optimal dispersers for degree~$d=1$ varieties, which are not constant on sets of size $2^{(\log{n})^c}$ for a~constant~$c$~\cite{li2016improved}. Also, the known constructions can handle large varieties of large degrees~\cite{R16}, or smaller varieties of size $2^{\alpha n}$ of constant degree (for a constant $\alpha$)~\cite{li2018improved}. On the other hand, the result of Cohen and Tal~\cite[Theorem 5]{CohenT15}, together with an efficient construction of affine dispersers from~\cite{li2016improved}, gives an explicit construction of $\left(16, \frac{n}{(\log{n})^c}, 2^{o(n)}\right)$-disperser (it handles varieties of the desired size, but only defined by fewer polynomials). Thus, although the currently known constructions do not suffice for proving new lower bounds, they are tantalizingly close to the ones needed for a simple proof of circuit lower bounds via Theorem~\ref{thm:transformations}.

We conclude this section with a~simple counting argument showing that a random function is a~disperser with great parameters.
\begin{restatable}{lemma}{random}
Let $d=d(n)$, $m=m(n)$, $s=s(n)$ be such that $s > 3dmn^d$. Then a~random function $f\colon\Ftwo^n\rightarrow\Ftwo$ is a $(d,m,s)$-disperser with probability $1-o(1)$.
\end{restatable}
\begin{proof}
Consider a~function~$f$ that is not a $(d,m,s)$-disperser. That is, $f$ is constant on some $(d,m)$-variety. In particular, $f$ can be uniquely specified by 
\begin{enumerate}
\item a $(d,m)$-variety $V$ where $f$ is constant,
\item one of the two possible constant values that $f$ takes on $V$,
\item values at the remaining (at most $2^n-s$) points.
\end{enumerate}

There are $k=\sum_{i=0}^{d}\binom{n}{i}\leq2dn^d$ monomials of degree at most~$d$ over $\{x_1, \dotsc, x_n\}$ (as any monomial is multilinear). Therefore, there are $2^k$ polynomials of degree at most~$d$, and at most $2^{mk}$ 
$(d,m)$-varieties.
Therefore, the number of functions $f$ which are not $(d,m,s)$-dispersers is bounded from above by
\begin{align*}
2^{mk}\cdot 2 \cdot 2^{2^n-s} \leq 2^{2dn^dm+1+2^n-s}\leq2^{2^n}\cdot o(1)
\end{align*}
Thus, a random function is an $(d,k,s)$-disperser with probability at least $1-o(1)$.
\end{proof}

\subsection{Correlation with Polynomials}
In this section we show that a~function that has small correlation with low-degree polynomials has high circuit complexity. We show this by using a known connection between correlation with polynomials and dispersers for varieties.

\begin{definition}
For two functions $f,g\colon\F^n\to\F$, we define their correlation as
\begin{align*}
\Cor(f,g) = \left| \Pr_{x}[f(x)=g(x)]-\Pr_{x}[f(x)\neq g(x)]\right| \, ,
\end{align*}
where $x$ is drawn uniformly at random from $\F^n$.
\end{definition}
By $\Cor(f,d)$ we denote the correlation of a~function~$f$ with polynomials of degree~$d$:
\begin{align*}
\Cor(f,d)=\max_{g}\Cor(f,g) \, ,
\end{align*}
where the maximum it taken over all polynomials $g$ of degree at most $d$.

There are several constructions of functions that have small correlation with polynomials of low degree~\cite{R87, S87, BNS92, VW08, D12, R16}, or sparse polynomials~\cite{V07}. In particular, the generalized inner product function has correlation $2^{-\Omega\left(\frac{n}{4^d\cdot d}\right)}$ with polynomials of degree~$d$~\cite{BNS92}, and Viola and Wigderson~\cite{VW08} constructed a function with correlation $2^{-\Omega\left(\frac{n}{2^d}\right)}$ with polynomials of degree $d$. See~\cite{V09} for an overview of the known bounds on correlation.

We use the fact that small correlation with polynomials of degree $d$ implies small correlation with products of polynomials of degree $d$, and, as a consequence, a disperser for varieties of degree~$d$.
\begin{restatable}[Implicit in \cite{D12,CT18,li2018improved}]{lemma}{fourier}\label{lem:fourier}
If $\Cor(f,d)\leq\eps$, then $f$ is $(d, \infty, \eps \cdot 2^n)$-disperser.
\end{restatable}
\begin{proof}
Consider a variety $V = \{x \in \Ftwo^n \colon q_1(x)=\dotsb=q_k(x)=0\}$, where each $q_i\colon\F^n\to\F$ is a non-constant polynomial of degree at most~$d$. Let $g(x)=\prod_{i=1}^k (q_i(x)\oplus1)$ be the indicator function of~$V$, and from the Fourier expansion we have 
\begin{align}
\label{eq:and fourier}
g(x) = \frac{\sum_{S\subseteq\{1,\ldots,k\}}(-1)^{\sum_{i\in S}q_i(x)}}{2^{k}} \, .
\end{align}
Now note that for any $S\subseteq\{1,\dotsc,k\}$, 
\begin{align*}
\left|\E_x\left[(-1)^{f(x)+\sum_{i\in S}q_i(x)}\right]\right|=\Cor\left(f,\sum_{i\in S}q_i(x)\right)\leq \eps \, ,
\end{align*}
 because $\sum_{i\in S}q_i(x)$ is a polynomial of degree at most $d$ and $\Cor(f,d)\leq \eps$.
Now
\begin{align*}
\left|\E_x\left[(-1)^{f(x)}\cdot g(x)\right]\right|&=\left|\E_x\left[(-1)^{f(x)}\cdot \frac{\sum_{S\subseteq\{1,\ldots,k\}}(-1)^{\sum_{i\in S}q_i(x)}}{2^{k}}\right]\right|\\
&=\frac{1}{2^{k}}\left|\E_x\left[\sum_{S\subseteq\{1,\ldots,k\}}(-1)^{f(x)+\sum_{i\in S}q_i(x)}\right] \right|\\
&\leq \frac{1}{2^{k}} \sum_{S\subseteq\{1,\ldots,k\}}\left|\E\left[(-1)^{f(x)+\sum_{i\in S}q_i(x)}\right] \right|\\
&\leq \frac{2^k \eps}{2^k}=\eps\, .
\end{align*}
In particular, for any variety $V$ of size $|V|>\eps 2^n$, $f(x)$ is not constant on $V$.
\end{proof}

Now Theorem~\ref{thm:dispersers} and Lemma~\ref{lem:fourier} imply the following result.\footnote{We remark that we do not apply these results to the depth reduction presented in this paper, but only to Valiant's depth reduction. Indeed, it would only give us a statement of the form: If $\Cor(f,16)\leq 2^{-n(1-\eps)}$, then $s(f) \ge 3.9(1-\eps)n-4$. But as we noted in Section~\ref{sec:motivating}, every Boolean function has correlation at least $2^{-n/2}$ with some linear polynomial.}
\begin{theorem}
Let $f \in B_n$ and $\eps>0$ be a~constant.
\begin{itemize}
\item If $\Cor(f,\omega(1))\leq 2^{-\eps n}$,
then $s_{\text{sp}}(f) = \omega(n)$.
\item If $\Cor(f,2^{(\log{n})^{1-o(1)}})\leq 2^{-\eps n}$,
then $s_{\log}(f)=\omega(n)$.
\item If $\Cor(f,n^{\eps})\leq 2^{-\omega(n/\log\log{n})}$,
then $s_{\log}(f)=\omega(n)$.
\end{itemize}
\end{theorem}

\subsection{Rigidity}
In order to prove super-linear circuit lower bounds for log-depth circuits via Valiant's reduction, one needs to construct matrices~$M$ with rigidity $\R_M\left(\frac{\delta n}{\log\log{n}}\right) > n^\eps$ or rigidity $\R_M(\eps n) > 2^{(\log{n})^{1-\delta}}$ for some constant $\eps>0$ and every constant $\delta>0$. For super-linear lower bounds for series-parallel circuits, one needs to find matrices with rigidity $\R_M\left(\eps n\right) > \delta$. Also, Razborov~\cite{R89} proved that rigidity $\R_M\left(2^{(\log\log{n})^c}\right) > \frac{n}{2^{(\log\log{n})^\eps}}$ for all $c \geq 1$ gives a language that does not belong to the polynomial hierarchy for communication complexity. The best known explicit lower bound on rigidity for every~$r$ is $\R(r)\geq\Omega\left(\frac{n}{r}\log{\frac{n}{r}}\right)$~\cite{friedman1993note,pudlak1991computation,shokrollahi1997remark,Lokam09}.\footnote{There is also a semi-explicit construction due to Goldreich and Tal~\cite{goldreich2016matrix}. This construction can be constructed in plain-exponential time $2^{O(n)}$ and has rigidity $\R(r)\geq\Omega\left(\frac{n^2}{r^2\log{n}}\right)$ for every $r\geq\sqrt{n}$. This bound is better than the known explicit bounds for $r=o\left(\frac{n}{\log{n}\log\log{n}}\right)$. It is also known~\cite{KV18} how to construct a matrix with rigidity as high as $\R(r)\geq\Omega(n)$ for any rank $r=n^{0.5-\eps}$ using subexponential time $2^{o(n)}$ .} Thus, for new bounds via Valiant's reduction (or Razborov's reduction for communication complexity), one needs to improve the known bounds asymptotically. 

In order to get new circuit lower bounds via Theorem~\ref{thm:linear}, we need to find a matrix $M\in\F^{n\times n}$ with rigidity $\R_M(0.75 n)>16$ (or a~rectangular matrix $M\in\F^{m\times n}$ for $m\geq n$ which is rigid for higher rank $\R_M(\frac{n}{2}+\frac{m}{4})>16$).
There are several explicit construction of matrices having rigidity $\R(\eps n)>16$ for some constant~$\eps$~\cite{friedman1993note,pudlak1991computation,shokrollahi1997remark,Lokam09}.  Valiant~\cite{V77} showed that a~random matrix $M\in\F^{n\times n}$ has rigidity $\R(r)\geq \frac{(n-r)^2-2n-\log{n}}{n\log(2n^2)}$ for any $r<n-\sqrt{2n+\log{n}}$. In particular, $\R_M(n-6\sqrt{n\log{n}})> 16$ for a~random matrix $M$. As for explicit constructions, Pudl{\'a}k and Vav{\v{r}}{\'\i}n~\cite{pudlak1991computation} found the exact value of rigidity (for every rank $r$) of the upper triangular matrix $T_n\in\F^{n\times n}$. In particular, they showed that $\R(\frac{n}{65})>16$. A matrix which is rigid for larger values of rank (at the price of having more outputs) was given in~\cite{pudlak1994some} and~\cite[Theorem~3.36]{jukna2013complexity}: A generator matrix $M\in\F^{m\times n}$ of a linear code with relative distance $\delta>0$ for any $r\leq n/16$ has rigidity 
\begin{align*}
\R_M(r)\geq \frac{\delta n \log(n/r)}{8(r+\log(n/r))} \, .
\end{align*}
We now show that using the ideas from~\cite{friedman1993note,shokrollahi1997remark}, one can improve this constant, but this is still not sufficient for getting new bounds using Theorem~\ref{thm:linear}. 

Recall that $\H(x)=-x\log{x}-(1-x)\log(1-x)$ for $0<x<1$, and that the generator matrix $M\in\F^{m\times n}$ of a code can always be transformed such that the first $n$ rows of $M$ form the identity matrix.

\begin{restatable}{lemma}{codes}\label{lem:rigidity}
Let $A\in\F^{(m-n)\times n}$, and let $I\in\F^{n\times n}$ be the identity matrix. If $M=
\begin{bmatrix}
I \\ A
\end{bmatrix}$ is a generator matrix of a linear code with relative distance $\delta$ and rate $R=\frac{n}{m}$, then $\R_A(r)>16$ for
\begin{align*}
r=\max_{0<\alpha<1} \left(\alpha n \cdot \H\left(\frac{\delta(1-\alpha)}{2\alpha(1-\alpha)R+32\alpha}\right)\right)-o(n) \, .
\end{align*}
\end{restatable}
\begin{proof}
We will show that for every $16$-sparse matrix $B$,
\begin{align*}
\rk(A\oplus B) > \alpha n \cdot \H\left(\frac{\delta(1-\alpha)}{2\alpha(1-\alpha)R+32\alpha}\right)-o(n) \,.
\end{align*}
First we take the $\alpha n$ sparsest columns of $B$. By Markov's inequality, each of them has at most $\frac{16m}{(1-\alpha)n}$ non-zero entries. Let $A',B',M'\in\F^{m\times \alpha n}$ be the submatrices of $A$, $B$, and $M$ corresponding to this set of $\alpha n$ columns. For a~vector $x\in\F^n$, let $|x|$ be the number of non-zero elements in it. 

Since $M$ generates a code with relative distance $\delta$, we have that for every non-zero $x\in\F^n$, $|Mx|\geq \delta m$. From $Mx=
\begin{bmatrix}
I \\ A
\end{bmatrix}
x=
\begin{bmatrix}
x \\ Ax
\end{bmatrix}
,\text{ we have that }
|Ax|\geq \delta m - |x|$.
Since this holds for every non-zero~$x$, including~$x$ with zeros in all coordinates \emph{not} in~$A'$, we get that for every $x\in\F^{\alpha n}$, $|A'x|\geq \delta m - |x|$.

Now we only consider non-zero $x\in\F^{\alpha n}$ with exactly $k=\beta n$ ones where $\beta=\frac{\delta(1-\alpha)}{(1-\alpha)R+16}-o(1)$. For such an $x$,
\begin{align*}
|(A'\oplus B')x|
\geq |A'x|-|B'x|
\geq \delta m - |x| - |x|\cdot\frac{16m}{(1-\alpha)n}
\geq \delta m - \beta n\left(1+\frac{16m}{(1-\alpha)n}\right)
>0
\end{align*}
due to the choice of $\beta$.
This implies that all linear combinations of exactly $k/2$ columns from $A'\oplus B'$ are distinct. That is, the columns of $A'\oplus B'$ span at least $\binom{\alpha n}{k/2}$ points in $\F^m$, and
\begin{align*}
\rk(A\oplus B)\geq \rk(A'\oplus B') &\geq \log\binom{\alpha n}{k/2} \\&=\alpha n \cdot \H(\beta/{2\alpha})-o(n)
\\&=\alpha n \cdot \H\left(\frac{\delta(1-\alpha)}{2\alpha(1-\alpha)R+32\alpha}\right)-o(n) \, .
\end{align*}
\end{proof}
Let us consider Justesen's code~\cite{justesen1972class},~\cite[Chapter~10, §11, Theorem~12]{macwilliams1977theory}. For $\delta=0.077$, we have an efficient construction of a linear code with rate $R=0.15$. In Lemma~\ref{lem:rigidity}, we set $\alpha=0.182$ and get that this matrix is rigid for rank $r>\frac{n}{64}$ beating the bound from~\cite{pudlak1991computation} (at the price of having $m-n=n(1/R-1)$ outputs).

If we take the concatenation of a Reed-Solomon code (as the outer code) and an optimal linear inner code, then for every $\delta$ we can construct in polynomial time a code with relative distance $\delta$ matching the Zyablov bound (see, e.g., the discussion in~\cite{alon1992construction}):
\begin{align*}
R=\max_{\delta\leq\mu\leq0.5}\left( \left(1-\H(\mu)\right)\left(1-\frac{\delta}{\mu}\right) \right) \, .
\end{align*}
In particular, if we take such a code with $\delta=0.49$, then in the Zyablov bound we set $\mu=0.493$ and get $R\approx 8 \cdot10^{-7}$. Now we set $\alpha=0.252$ in Lemma~\ref{lem:rigidity}, and get rigidity for rank as high as $r>\frac{n}{15}$ (at the price of having too many outputs).

\subsection{Open Problems}
We conclude with a short summary of pseudorandom objects which would lead to new circuit lower bounds via depth reductions described in Section~\ref{sec:reductions}.

\begin{openproblem}
Prove that E$^{\text{NP}}$ contains a language $f$ having one of the following properties:
\begin{itemize}
\item $f$ cannot be computed by an $\dm{2^{0.8n}}{n\cdot 2^{15}}{16}$.
\item $f$ is a disperser for varieties of size at least $2^{0.2n}$ defined by $1.05n$ polynomials each of which depends on at most $16$ variables (and, thus, has degree at most $16$).
\item $f$ is a linear function defined by a matrix $M\in\F^{n\times n}$ of rigidity $\R_M(0.8n)>16$ (that is, in order to decrease the rank of $M$ to $0.8n$, one has to change more than $16$ elements in some row of $M$).
\end{itemize}
\end{openproblem}

\begin{openproblem}
Show that every DeMorgan formula of size $s$ has a probabilistic polynomial over $\F_2$ of degree $s^{0.99}$ and error $1/3$, or give evidence this is not true. We conjecture the degree can be made $O(\sqrt{s})$.
\end{openproblem}


\subsection*{Acknowledgement}
We thank Navid Talebanfard and Emanuele Viola for helpful discussions.

\bibliographystyle{alpha}
\bibliography{refs}

\end{document}